\documentclass[twocolumn,twoside]{IEEEtran}

\pdfpageattr {/Group << /S /Transparency /I true /CS /DeviceRGB>>} 

\usepackage{amsmath}
\usepackage{amssymb}
\usepackage{amsthm}
\usepackage{latexsym}
\usepackage{verbatim}
\usepackage{graphicx}

\newtheorem{thm}{Theorem}
\newtheorem{prop}{Proposition}
\newtheorem{lem}{Lemma}

{\begin{list}               
    {$\bullet$ \hfill}{
        \setlength{\leftmargin}{\parindent}
        \setlength{\parsep}{0.04\baselineskip}
        \setlength{\itemsep}{0.5\parsep}
        \setlength{\labelwidth}{\leftmargin}
        \setlength{\labelsep}{0em}}
    }
{\end{list}}

\providecommand{\eref}[1]{\eqref{eq:#1}}  
\providecommand{\cref}[1]{Chapter~\ref{chap:#1}}
\providecommand{\sref}[1]{Section~\ref{sec:#1}}

\providecommand{\R}{\ensuremath{\mathbb{R}}}
\providecommand{\C}{\ensuremath{\mathbb{C}}}

\renewcommand{\S}{\ensuremath{\mathbb{S}}}
\providecommand{\Z}{\ensuremath{\mathbb{Z}}}

\providecommand{\abs}[1]{{\left|#1\right|}}
\providecommand{\norm}[1]{\left\lVert#1\right\rVert}
\providecommand{\inprod}[1]{\left\langle#1\right\rangle}
\providecommand{\set}[1]{\left\{#1\right\}}

\providecommand{\bydef}{\overset{\text{def}}{=}}
\providecommand{\diag}{\mathop{\mathrm{diag}}}

\providecommand{\parder}[2]{\frac{\partial{#1}}{\partial{#2}}}

\providecommand{\di}{\ensuremath{~\text{d}}}

\providecommand{\I}{\ensuremath{\mathrm{j}}}

\renewcommand{\vec}[1]{\ensuremath{\bm{#1}}}
\providecommand{\mat}[1]{\ensuremath{\bm{#1}}}
\providecommand{\wh}[1]{\ensuremath{\widehat{#1}}}
\providecommand{\wt}[1]{\ensuremath{\widetilde{#1}}}

\providecommand{\mA}{\mat{A}} \providecommand{\mB}{\mat{B}}

\providecommand{\mI}{\mat{I}}

 \providecommand{\mR}{\mat{R}}
 \providecommand{\mU}{\mat{U}} 
\providecommand{\mV}{\mat{V}} \providecommand{\mT}{\mat{T}}

 \providecommand{\mG}{\mat{G}}

\providecommand{\mX}{\mat{X}}\providecommand{\mY}{\mat{Y}}
\providecommand{\mZ}{\mat{Z}}

 \providecommand{\vb}{\vec{b}}
\providecommand{\vc}{\vec{c}} \providecommand{\vd}{\vec{d}}
\providecommand{\ve}{\vec{e}} \providecommand{\vf}{\vec{f}}
 
\providecommand{\vg}{\vec{g}}
\providecommand{\vh}{\vec{h}}

 \providecommand{\vr}{\vec{r}}
\providecommand{\vs}{\vec{s}}

\providecommand{\vx}{\vec{x}} 
 
 \providecommand{\vzero}{\vec{0}}
 \providecommand{\vv}{\vec{v}}

\providecommand{\conv}{\ast}

\usepackage{bm}
\usepackage[noadjust]{cite}
\usepackage{algorithm}
\usepackage{algpseudocode}
\usepackage{xcolor}
\usepackage{mathabx}
\usepackage{placeins}
\usepackage{hyperref}
\usepackage{mathtools}
\usepackage[percent]{overpic}

\newcommand{\shf}{\wh{f}}
\newcommand{\shg}{\wh{g}}
\newcommand{\shh}{\wh{h}}

\newcommand{\shv}{\wh{v}}
\newcommand{\shy}{\wh{y}}
\newcommand{\trace}{\mathrm{trace}}
\newcommand{\mXi}{\mat{\Xi}}

\newcommand{\T}{\top}
\renewcommand{\H}{H}
\newcommand{\SO}{\ensuremath{\mathbb{SO}}}
\newcommand{\supp}{\ensuremath{\mathop{\mathrm{supp}}}}

\newcommand{\mDelta}{\bm{\Delta}}
\newcommand{\expect}{\mathbb{E}}

\begin{document}
\title{Sampling Sparse Signals on the Sphere:\\Algorithms and Applications}

\author{Ivan~Dokmani\'c,~\IEEEmembership{Student Member,~IEEE,}
        and~Yue~M.~Lu,~\IEEEmembership{Senior Member,~IEEE}
\thanks{ I. Dokmani\'{c} is with the School of Computer and Communication
Sciences, Ecole Polytechnique F\'{e}d\'{e}rale de Lausanne (EPFL), CH-1015
Lausanne, Switzerland (e-mail: ivan.dokmanic@epfl.ch@epfl.ch). His research
was supported by an ERC Advanced Grant---Support for Frontier
Research---SPARSAM Nr: 247006, and a Google PhD Fellowship.}%
\thanks{Y. M. Lu is with the School of Engineering and Applied Sciences, Harvard University, Cambridge, MA 02138 USA (e-mail: yuelu@seas.harvard.edu). He was supported in part by the U.S. National Science Foundation under grant CCF-1319140.}%
\thanks{Preliminary material in this paper will be presented at the 40th IEEE International Conference on Acoustics, Speech and Signal Processing (ICASSP), 19--24 April 2015, in Brisbane, Australia.}}%

\markboth{}%
{Dokmani\'c and Lu: Sampling Sparse Signals on the Sphere}

\maketitle

\begin{abstract}
  We propose a sampling scheme that can perfectly reconstruct a collection of
  spikes on the sphere from samples of their lowpass-filtered observations.
  Central to our algorithm is a generalization of the annihilating filter
  method, a tool widely used in array signal processing and
  finite-rate-of-innovation (FRI) sampling. The proposed algorithm can
  reconstruct $K$ spikes from $(K+\sqrt{K})^2$ spatial samples. This sampling
  requirement improves over previously known FRI sampling schemes on the
  sphere by a factor of four for large $K$.

  We showcase the versatility of the proposed algorithm by applying it to
  three different problems: 1) sampling diffusion processes induced by
  localized sources on the sphere, 2) shot noise removal, and 3) sound source
  localization (SSL) by a spherical microphone array. In particular, we show
  how SSL can be reformulated as a spherical sparse sampling problem.
\end{abstract}

\begin{IEEEkeywords}
  Sphere, sparse sampling, diffusion sampling, sound source localization,
  annihilation filter, finite rate of innovavtion, shot noise removal,
  spherical harmonics
\end{IEEEkeywords}


\section{Introduction}

\IEEEPARstart{N}{umerous} signals live on a sphere. Take, for example, any
signal defined on Earth's surface
\cite{Evans:1998ii,Simons:2006gg,Audet:2011ei}. Signals from space measured on
Earth
\cite{Jarosik:2011dy,MacPhie:1975bj} also have a spherical domain. In
acoustics, spherical microphone arrays output a time-varying signal supported
on a sphere \cite{Meyer:2002bl,Jarrett:2012fi}, while in diffusion weighted
magnetic resonance imaging fiber orientations live on a sphere
\cite{Tournier:2004bz}. In practice, we only have access to a finite number of
samples of such signals. Thus, sampling and reconstruction of spherical
signals is an important problem.

Just as signals in Euclidean domains can be expanded via sines and cosines,
one can naturally represent spherical signals in the Fourier domain via
\emph{spherical harmonics} \cite{Driscoll:1994bp}. A signal is
\emph{bandlimited} if
it is a linear combination of finitely many spherical harmonics. Sampling
bandlimited signals on the sphere has been studied extensively: for signals
bandlimited to spherical harmonic degree $L$, Driscoll and Healy
\cite{Driscoll:1994bp} proposed a sampling theorem that requires $4L^2$
spherical samples. The best exact general purpose sampling theorem due to
McEwen and Wiaux uses $2L^2$ samples
\cite{McEwen:2011ib}. Recently, Khalid, Kennedy and McEwen devised a stable
sampling scheme that requires the optimal number of samples, $L^2$
\cite{Khalid:2014is}.

In this paper, we study the problem of sampling  localized spikes on the
sphere; in the limit, the spikes become Dirac delta functions. Such sparse
signals on the sphere are encountered in many problems. For example, various
acoustic sources are well-approximated by point sources; the directional
distribution of multiple sources is then a finite collection of spikes. Stars
in the sky observed from Earth are angular spikes, and so are plume sources on
Earth.

Localized spikes are not bandlimited, so the bandlimited sampling theorems
\cite{Driscoll:1994bp,McEwen:2011ib,Khalid:2014is} do not apply. In this
paper, we propose an algorithm to perfectly reconstruct collections of spikes
from their lowpass-filtered observations. Our algorithm efficiently
reconstructs $K$ spikes when the bandwidth of the lowpass filter is at least
$K + \sqrt{K}$.

\subsection{Prior Art}

Our work is in the same spirit as finite rate-of-innovation (FRI) sampling,
introduced by Vetterli, Marziliano, and Blu \cite{Vetterli:2002bs}. They
showed that a stream of $K$ Diracs on the line can be efficiently recovered
from $2K+1$ samples. Initially developed for 1D signals, the original FRI sampling was extended to
2D and higher-dimensional signals in \cite{Maravic:2004bz,Shukla:ca}, and its
performance was studied in noisy conditions \cite{Maravic:2005fd,Dragotti:2009es}.

In a related work
\cite{DeslauriersGauthier:2013ij,DeslauriersGauthier:2013be},
Deslauriers-Gauthier and Marziliano proposed an FRI sampling scheme for
signals on the sphere, reconstructing $K$ Diracs from $4K^2$ samples. Their
motivating application is the recovery of the fiber orientations in diffusion
weighted magnetic resonance imaging
\cite{DeslauriersGauthier:2012if, Tournier:2004bz}. They further show that if
only $3K$ spectral bins are active, the required number of samples can be
reduced to $3K$. Sampling at this lower rate, however, relies on the
assumption that we can apply arbitrary spectral filters to the signal before
sampling. This is known as spatial anti-aliasing---a procedure that is
generally challenging or impossible to implement in most applications
involving spherical signals, where we only have access to finite samples of the
underlying continuous signals.\footnote{This is not to be confused with
spatial anti-aliasing in image downsampling, where we do have access to all
pixels.}

In many applications, the sampling kernels (\emph{i.e.}, the lowpass
filters) through which we observe the spikes are provided by some underlying
physical process (\emph{e.g.}, point spread functions and Green's functions).
These kernels are often approximately bandlimited, but we cannot further
control or design the spectral selectivity of these kernels. This
impossibility of arbitrary spatial filtering suggests that our goal is to
reduce the required bandwidth, or more practically, to maximize the number of
spikes that we can reconstruct at a given bandwidth.

Recently, Bendory, Dekel and Feuer proposed a spherical super-resolution
method \cite{Bendory:2014tl, Bendory:2014gp}, extending the results of
Cand\`es and Fernandez-Granda \cite{Candes:2014br} to the spherical domain.
They showed that an ensemble of Diracs on the sphere can be reconstructed from
projections onto a set of spherical harmonics by solving a semidefinite
program, provided that the Diracs satisfy a minimal separation condition. When
the Diracs are constrained to a discrete set of locations, their formulation
allows them to bound the recovery error in the presence of noise. Our
non-iterative (thus very fast) algorithm based on FRI does not require any
separation between the Diracs. We also allow the weights to be complex, which
may be important in applications (for an example on sound source localization,
see Section \ref{sub:ssl}). On the other hand, we need to assume that the
number of Diracs is known a priori (or that it can be estimated through other
means), whereas in \cite{Bendory:2014gp,Bendory:2014tl} no such assumption is
necessary.

\subsection{Outline and Main Contributions}

We start by reviewing some basic notions of harmonic analysis on the sphere in
Section \ref{sec:sph}. We then present the main result of this work in Section
\ref{sec:spherefri}\,: A collection of $K$ Diracs on the sphere can be
reconstructed from its lowpass filtered version, provided that the bandwidth
of the sampling kernel is at least $K + \sqrt{K}$. This bandwidth requirement
also implies that $(K+\sqrt{K})^2$ spatial samples taken at generic locations
suffice to reconstruct the $K$ Diracs. We establish this result by
constructing a new algorithm for spherical FRI sampling. Compared to $4K^2$
samples as required in a previous work \cite{DeslauriersGauthier:2013ij}, our
algorithm reduces the numbers of samples via a more efficient use of the
available spectrum. For large $K$, the required number of samples is reduced by
a factor of up to 4. The proposed algorithm is first developed for the
noiseless case. Procedures to improve the robustness of the algorithm in noisy
situations are presented in Section~\ref{sub:denoising_str}, and we compare the
performance of the algorithm with the Cram\'er-Rao lower bound
\cite{Kay:1998wba} in Section \ref{sub:crlb}. Section~\ref{sec:applications}
presents the applications of the proposed algorithm to three problems: 1)
sampling diffusion processes on the sphere, 2) shot noise removal, and 3)
sound source localization. These diverse applications demonstrate the
usefulness and versatility of our results. We conclude in \sref{conclusion}.

This paper follows the philosophy of reproducible research. All the results
and examples presented in the paper can be reproduced using the code available
at \url{http://lcav.epfl.ch/ivan.dokmanic}.


\section{Harmonic Analysis on the Sphere and \\ Problem Formulation}

\label{sec:sph}

\subsection{Spherical Harmonics}

We briefly recall the definitions of spherical harmonics and spherical
convolution. The 2-sphere is defined as the locus of points in $\R^3$ with unit
norm,
\begin{equation*}
  \S^2 \bydef \set{\vx \in \R^3~|~\vx^\T \vx = 1}.
\end{equation*}

In what follows, we often use $\xi$ to represent a generic point on the
sphere. In addition to the standard Euclidean representation $\xi = [x, \ y, \
z]^\T$, points on $\S^2$ can also be conveniently parameterized by angles of
colatitude and azimuth, \emph{i.e.}, $\xi = (\theta, \phi)$, with $\theta$
measured from the positive $z$-axis, and $\phi$ measured in the $xy$ plane
from the positive $x$-axis. The two equivalent representations are related by
the following conversion,
\begin{equation}
\begin{aligned}
  \label{eq:polar_conversion}
  x &= \sin(\theta) \cos(\phi), \\
  y &= \sin(\theta) \sin(\phi), \\
  z &= \cos(\theta).
\end{aligned}
\end{equation}

The Hilbert space of square-integrable functions on the sphere, $L^2(\S^2)$, is
defined through the corresponding inner product. For two functions $f, g \in L^2(\S^2)$ we
have
\begin{equation}
  \label{eq:inner_product}
  \inprod{f,\ g} \bydef \int_{\S^2} f(\xi) \overline{g(\xi)} \di \xi,
\end{equation}
where $\di \xi = \sin(\theta) \di \theta \di \phi$ is the usual
rotationally invariant measure on the sphere. With respect to this inner
product, spherical harmonics form a natural orthonormal Fourier basis for
$L^2(\S^2)$. They are defined as \cite{Driscoll:1994bp}
\begin{equation}
  \label{eq:sph}
  Y_{\ell}^m(\theta, \phi) = N_{\ell}^m P_{\ell}^{\abs{m}}(\cos \theta) e^{\I m \phi},
\end{equation}
where the normalization constant is
\begin{equation}
  \label{eq:norm}
  N_{\ell}^m = (-1)^{(m + \abs{m})/2} \sqrt{\frac{(2\ell+1)}{4\pi} \frac{(l-\abs{m})!}{(l+\abs{m})!}},
\end{equation}
and $P_{\ell}^m(x) $ is the associated Legendre polynomial of degree $\ell$
and order $m$. Note that different communities sometimes use different
normalizations and sign conventions in the definitions of spherical harmonics
and associated Legendre polynomials. As long as applied consistently, the
choice of convention does not affect our results.\footnote{It is common to
write the spherical harmonic order $m$ in the superscript. We will keep this
convention for the associated Legendre polynomials $P_{\ell}^{\abs{m}}$,
spherical harmonics $Y_{\ell}^m$, normalization constants $N_\ell^m$ and the
spherical Fourier coefficients $\shf_\ell^m$. It is not to be confused with
integer powers such as $x^\ell$.}

In this paper, we adopt the following definition
\begin{equation}
  \label{eq:associated_legendre} P_{\ell}^m(x) \bydef (-1)^m (1-x^2)^{m/2}
  \frac{d^m}{dx^m}P_{\ell}(x), \ \text{for~} m \geq 0,
\end{equation}
where $P_{\ell}(x)$ is the Legendre polynomial of degree $\ell$
\cite{Abramowitz:1972vwa}.

Any square integrable function on the sphere, $f \in L^2(\S^2)$, can be expanded
in the spherical harmonic basis,
\begin{equation}
  \label{eq:ExpansionSH}
  f(\theta, \phi) = \sum_{\ell=0}^{\infty} \sum_{\abs{m} \leq \ell} \shf_{\ell}^m Y_{\ell}^m (\theta, \phi).
\end{equation}
The Fourier coefficients are computed as
\begin{equation}\label{eq:Fourier_transform}
  \shf_{\ell}^m = \inprod{f, Y_{\ell}^m} = \int_{\S^2} f(\xi) \overline{Y_{\ell}^m(\xi)} \di \xi.
\end{equation}
The coefficients $\big[ \shf_{\ell}^m, (\ell, m) \in \mathcal{I}  \big]$ form a
countable set supported on an infinite triangle of indices,
\begin{equation}
  \label{eq:triangle}
  \mathcal{I} = \set{(\ell, m) \in \Z^2 \ | \ \ell \geq 0, \abs{m} \leq \ell}.
\end{equation}
We say that $f$ is \emph{bandlimited} with bandwidth $L$ if $\shf_\ell^m = 0$ for $\ell
\geq L$. Often we think of $L$ as the smallest integer such that this
holds. For a bandlimited function, the triangle $\mathcal{I}$ is cut off at
$\ell = L$. In what follows, we use
\begin{equation}
\label{eq:triangle_L}
\mathcal{I}_L \bydef \set{(\ell, m) \in \Z^2 \ | \ 0 \le \ell < L, \abs{m} \leq \ell}
\end{equation}
to represent the spectral support of a bandlimited function with bandwidth
$L$. The set $\mathcal{I}_L$ contains $L^2$ indices, so we can represent the
spectrum as an $L^2$-dimensional column vector
\begin{equation}
  \label{eq:vectorized_spectrum}
  \wh{\vf} \bydef \left[\shf_0^{\; 0},\ \shf_1^{-1},\ \shf_1^{\; 0}, \ \shf_1^{\; 1},\ \ldots, \ \shf_{L-1}^{-L+1},\ \ldots, \ \shf_{L-1}^{\; L-1}\right]^\T.
\end{equation}

\subsection{Rotations and Convolutions on the Sphere}

Let $\SO_3$ denote the group of rotations in $\R^3$; any rotation $\varrho
\in \SO_3$ is parameterized by three angles that specify rotations about three
distinct axes. Thus we can write $\varrho =
\varrho(\alpha,\beta,\gamma)$. The commonest parameterization is called
\emph{Euler angles} \cite{Varshalovich:1989ul}.

Counter-clockwise rotation of a vector $\vx \in \R^3$ about the $z$-axis is
achieved by multiplying $\vx$ by the corresponding rotation matrix,
\begin{equation*}
  \mR_z(\alpha) = \begin{bmatrix}
  \cos \alpha & -\sin \alpha & 0 \\
  \sin \alpha & \cos \alpha & 0 \\
  0 & 0 & 1 
  \end{bmatrix},
\end{equation*}
where $\alpha$ is the rotation angle. Rotation matrices around axes $x$ and
$y$ can be defined analogously.

We use $\Lambda(\varrho)$ to represent the rotation operator corresponding to
$\varrho$, that acts on spherical functions. Thus for $f$ a function on the
sphere, $\Lambda(\varrho) f$ represents the rotated function, defined as
\begin{equation}
  [\Lambda (\varrho) f] (\xi) \bydef f(\varrho^{-1} \circ \xi),
\end{equation} where $\rho^{-1}$ is the inverse rotation of $\rho$, and by
$\varrho^{-1} \circ \xi$ we mean pre-multiplying by $\mR(\varrho^{-1})$ the
unit column vector corresponding to $\xi$,
\emph{cf.} \eqref{eq:polar_conversion}. Compare this definition with the
Euclidean case where shifting the argument to the left (subtracting a positive
number) results in the shift of the function to the right.

There are various definitions of convolution on the sphere, all being
non-commutative. One function, call it $f$, provides the weighting for the
rotations of the other function $h$. A standard definition is then
\cite{Driscoll:1994bp,Dokmanic:2010du}
\begin{equation}
\label{eq:convolution}
\begin{aligned}
  \left[ f \conv h \right] (\xi) 
  & \bydef \left[ \left( \frac{1}{2\pi} \int_{\SO_3} \di \varrho \cdot f(\varrho \circ \eta) \cdot \Lambda(\varrho) \right) h \right] (\xi) \\
  &= \frac{1}{2\pi} \int_{\SO_3} f(\varrho \circ \eta) h(\varrho^{-1} \circ \xi) \di \varrho,
\end{aligned}
\end{equation}
where $\eta \in \S^2$ is the north pole. It is easy to verify that this
definition generalizes the standard convolution in Euclidean spaces, with the
rotation operator $\varrho$ playing the same role as translations do on the
line. Because the spherical convolution is not commutative, it is important to
fix the ordering of the arguments. In our case, the second argument---$h$ in
\eqref{eq:convolution}---will always be the filter, \emph{i.e.}, the
observation kernel.

The familiar convolution--multiplication rule in standard Euclidean domains
holds for spherical convolutions too. It can be shown \cite[Theorem
1]{Driscoll:1994bp} that for any two functions $f, h \in L^2(\S^2)$, the
Fourier transform of their convolution is a pointwise product of the
transforms, \emph{i.e.},
\begin{equation}
  \label{eq:ConvSpectrum}
  (\widehat{f \conv h})_\ell^m = \sqrt{\frac{4\pi}{2 \ell + 1}} \ \shf_\ell^m \ \shh_\ell^{0}.
\end{equation}

We note that $f$ can also be a generalized function (a distribution). In
particular, we consider spherical Dirac delta functions, defined as \cite{Duffy:2001uy}
\begin{equation}\label{eq:delta}
  \delta(\theta, \phi; \theta_0, \phi_0) = \frac{\delta(\theta - \theta_0) \delta(\phi - \phi_0)}{\sin(\theta)},
\end{equation}
and weighted sums of Dirac deltas. To lighten the notation, we often write
$\delta(\xi; \xi_0)$. With the definition in \eqref{eq:delta}, it is ensured
that
\begin{equation}
  \int_{\S^2} \delta(\xi; \xi_0) \di \xi = 1, \ \forall \xi_0 \in \S^2.
\end{equation}

\subsection{Problem Formulation}

Consider a collection of $K$ Diracs on the sphere
\begin{equation}
  \label{eq:model}
  f(\xi) = \sum_{k=1}^{K} \alpha_k \delta(\xi; \xi_k),
\end{equation}
where the weights $\set{\alpha_k \in \C}_{k=1}^K$ and the locations of the Diracs
$\set{\xi_k = (\theta_k, \phi_k)}_{k=1}^K$ are all unknown parameters. Let $y(\xi)$ be
a filtered version of $f(\xi)$, \emph{i.e.},
\[
y(\xi) = [f \ast h](\xi),
\]
where the filter (or sampling kernel) $h(\xi)$ is a bandlimited function with bandwidth $L$. We further assume that the spherical Fourier transform of $h(\xi)$ is nonzero within its spectral support, \emph{i.e.}, $\widehat{h}_{\ell}^m \neq 0$ for all $\ell < L$. Given spatial samples of $y(\xi)$, we would like to reconstruct $f(\xi)$, or equivalently, to recover the unknown parameters $\set{(c_k, \xi_k)}_{k=1}^K$.

Since the filtered signal $y(\xi)$ is bandlimited, we can use bandlimited sampling theorems on the sphere (\emph{e.g.}, \cite{Driscoll:1994bp,McEwen:2011ib}) or direct linear inversion (see Section \ref{sub:from_samples_to_spectrum}) to recover its Fourier spectrum $\widehat{y}_{\ell}^m$ from its spatial samples of sufficient density. Using the convolution-multiplication identity in \eref{ConvSpectrum}, we can then recover the lowpass subband of $f(\xi)$ as
\[
\shf_{\ell}^m = \left[ (2\ell + 1)/(4\pi) \right]^{1/2} \cdot \left( \shy_{\ell}^{\,m} \, / \, \shh_{\ell}^{0} \right),
\]
for $0 \le \ell < L$ and $\abs{m} \le \ell$. Being a collection of Diracs, $f \notin L^2(\S^2)$, but its Fourier transform $\shf_{\ell}^m$ can still be computed via \eref{Fourier_transform} in the sense of distributions as

\begin{equation}
\begin{aligned}
  \label{eq:DiracSpectrum}
  \shf_{\ell}^m &= \sum_{k=1}^K \alpha_k \overline{Y_{\ell}^m(\theta_k, \phi_k)}\\
  	&= N_{\ell}^m \sum_{k=1}^K \alpha_k  P_{\ell}^{\abs{m}}(\cos \theta_k) e^{-\I m \phi_k}.
\end{aligned}
\end{equation}

The problems we address in this paper can now be stated as follows: Can we
reconstruct a collection of $K$ Diracs on the sphere from its Fourier
coefficients $\shf_{\ell}^m$ in the lowpass subband $\mathcal{I}_L$ as defined
in \eref{triangle_L}? If so, then what is the minimum bandwidth $L$ that
allows us to do so? In practice, the sampling kernel is often given and not
subject to our control. In this case, the previous question can be
reformulated as determining the maximum number of spikes that we can
reconstruct at a given bandwidth $L$.


\section{Sampling Spherical FRI Signals}

\label{sec:spherefri}

In this section we address the questions stated above.  Our main result can be
summarized in the following theorem:
\begin{thm}
  \label{thm:sfri}
  Let $f$ be a collection of $K$ Diracs on the sphere $\S^2$, with complex
  weights $\set{\alpha_k}_{k=1}^K$ at locations $\set{\xi_k = (\theta_k,
  \phi_k)}_{k=1}^{K}$, as in \eqref{eq:model}. Convolve $f$ with a
  bandlimited sampling kernel $h_L$, where the bandwidth $L \ge  K + \sqrt{K}$,
  and sample the resulting signal $[f \conv h_L](\xi)$ at $L^2$ points
  $\set{\psi_n \in \S^2}_{n=1}^{L^2}$ chosen uniformly at random on $\S^2$. Then almost
  surely the samples
  \begin{equation*}
    f_n = [f \conv h_L](\psi_n), \quad n = 1, \ldots, L^2
  \end{equation*}
  are a sufficient characterization of $f$.
\end{thm}

We provide a constructive proof of this theorem by presenting an algorithm
that can efficiently recover $K$ localized spikes from $L^2$ samples, where $L
\geq K + \sqrt{K}$. Before presenting the algorithm and the proof, we first
define some relevant notation and state two lemmas.

\subsection{From Samples to the Fourier Transform}
\label{sub:from_samples_to_spectrum}

Our algorithms perform computation with spectral coefficients. In practice, we
have access to spatial samples of the function, so we need a procedure to
convert between the spatial and the Fourier representations. We first describe
a method to compute the Fourier transform from samples taken at generically
placed sampling points.

Let the function $f \in L^2(\S^2)$ have bandwidth $L$; then we can express it
as
\begin{equation}
  \label{eq:SHExpansion}
  f(\theta, \phi) = \sum_{\ell = 0}^{L-1} \sum_{m = -\ell}^{\ell} \shf_{\ell}^m Y_{\ell}^m(\theta, \phi).
\end{equation}
Choose a set of sampling points $\set{\psi_n \in \S^2}_{n = 1}^{N}$, and let
$\mY = \left[y_{n, (\ell, m)}\right]$ where $y_{n, (\ell, m)} =
Y_{\ell}^m(\psi_n)$. Furthermore, let $\vf = [f(\psi_1), \ \ldots, \
f(\psi_n)]^\T$ be the vector of samples of $f$. We can then write
\begin{equation}
  \vf = \mY \wh{\vf},
\end{equation}
where $\wh{\vf}$ is the $L^2$-dimensional vector of spectral coefficients as
defined in \eqref{eq:vectorized_spectrum}. The goal is to recover the spectral
coefficients $\wh{\vf}$. We can recover $\wh{\vf}$ from $\vf$ as soon as the
matrix $\mY$ has full column rank. In that case, we compute
\begin{equation}
  \label{eq:pinv_spectrum_computation}
  \wh{\vf} = \mY^{\dagger} \vf,
\end{equation}
where $\mY^{\dagger}$ denotes the Moore-Penrose pseudoinverse of the matrix
$\mY$.

In particular, if we draw the samples uniformly at random on the sphere, we
can show that $\mY$ is regular with probability one:

\begin{prop}
  \label{prop:random_sampling}
  Draw $N$ sampling points from any absolutely continuous probability measure
  on the sphere (\emph{e.g.} uniformly at random). Then $\mY$ has full column
  rank almost surely if $N \geq L^2$, that is, if it has at least as many rows
  as columns.
\end{prop}

The proof of this proposition is identical to that of Theorem 3.2 in
\cite{Bass:2005dr}, and is thus omitted.

The above result indicates that we can recover the spectral coefficients
$\shf_{\ell}^m$ in the lowpass region $\mathcal{I}_L$ from $L^2$ samples taken
at generic points on the sphere. The reconstruction requires a matrix
inversion as in \eqref{eq:pinv_spectrum_computation}.

Much faster reconstruction is possible when the function is sampled on certain
regular grids. In that case, we can leverage the structure of $\mY$ to
accelerate the matrix inversion. Such efficient schemes were proposed by
Driscoll and Healy \cite{Driscoll:1994bp}, requiring $4L^2$ samples; by McEwen
and Wiaux \cite{McEwen:2011ib}, requiring $2L^2$ samples; and most recently,
by Khalid, Kennedy and McEwen \cite{Khalid:2014is}, requiring $L^2$ samples.

\subsection{The Data Matrix} 
\label{sub:the_data_matrix}

Using the definition of associated Legendre polynomials in
\eqref{eq:associated_legendre}, we rewrite the spherical harmonics
\eqref{eq:sph} as

\begin{equation}
  \label{eq:rewrite}
  Y_{\ell}^m(\theta, \phi) = \wt{N}_{\ell}^m (\sin \theta)^{\abs{m}} \left[ \frac{d^{\abs{m}}}{d(\cos \theta)^{\abs{m}}} P_{\ell}(\cos \theta) \right] e^{\I m \phi},
\end{equation}
where $\wt{N}_{\ell}^m = (-1)^m N_{\ell}^m$.

The essential observation is that the bracketed term in \eqref{eq:rewrite} is
a polynomial in $x =
\cos \theta$. At bandwidth $L$, the largest spherical harmonic degree is
$L-1$, so the largest power of $x$ in \eqref{eq:rewrite} is $L-1$ as well. It
follows that we can rewrite the derivative term as a linear combination of
powers of $x$, \emph{i.e.}
\begin{equation}
  \label{eq:polyinprod}
  \wt{N}_\ell^m \frac{d^{\abs{m}}}{d (\cos\theta)^{\abs{m}}} P_\ell(\cos \theta) = \vc_{\ell m}^\T \vx,
\end{equation}
where $\vx \bydef [x^{L-1}, \ x^{L-2}, \ \cdots, \ x, \ 1]^\T$, $x = \cos
\theta$ and $\vc_{\ell m} \in \R^L$ contains the corresponding polynomial
coefficients.

Using the dot-product formulation \eqref{eq:polyinprod}, the spectrum of $f$,
as given by \eqref{eq:DiracSpectrum}, can be expressed as
\begin{equation}
  \label{eq:dot_product_prepared}
  \shf_{\ell}^m = \vc_{\ell m}^\T \sum_{k=1}^K \alpha_k \vx_{k} (\sin \theta_k)^{\abs{m}} e^{-\I m \phi_k},
\end{equation}
where $\vx_k \bydef [x^{L-1}_k, \ x^{L-2}_k, \ \cdots, \ x_k, \ 1]^\T$ with $x_k = \cos \theta_k$, and we factored
$\vc_{\ell m}^\T$ out of the summation as it does not depend on $k$. 

A key ingredient in our proposed algorithm is what we call the \emph{data
matrix} $\mDelta$, formed as a product of three matrices,
\begin{equation}
  \label{eq:data_matrix}
  \mDelta \bydef \mX \mA \mU,
\end{equation}
where
\begin{equation}
  \label{eq:matrix_X}
  \mX = [\vx_{1}, \ \cdots, \ \vx_{K}] \in R^{L \times K},
\end{equation}
is a Vandermonde matrix with roots $\cos \theta_k$, $\mA =
\diag(\alpha_1, \ldots, \alpha_K)$ is the diagonal matrix of Dirac magnitudes,
and we define
\begin{equation}
  \label{eq:matrix_U}
  \mU = [u_{km}] \in \R^{K \times (2L-1)},
\end{equation}
with $u_{km} \bydef (\sin \theta_k)^{\abs{m}} e^{-\I m\phi_k}$.

It is convenient to keep a non-standard indexing scheme for the rows and
columns of $\mDelta$, as illustrated in
Fig.~\ref{fig:data_matrix_illustration}B. Rows of $\mDelta$, indexed by $p$,
correspond to decreasing powers of $\cos \theta_k$, from $p = L-1$ at the top, to
$p = 0$ at the bottom; columns correspond to $u_{km}$, with $m$ increasing from
$-L+1$ on the left, to $L-1$ on the right. We see from
\eqref{eq:dot_product_prepared} and
\eqref{eq:data_matrix} that computing any spectral coefficient $\shf_\ell^m$
amounts to applying a linear functional on $\mDelta$ as follows
\begin{equation}\
  \label{eq:linfunc}
  \shf_{\ell}^m = \vc_{\ell m}^\T  \mDelta \ve_m = \inprod{\vc_{\ell m} \ve_m^\T,\ \mDelta}_F,
\end{equation}
where $\ve_m \in \R^{2L-1}$ is the vector with one in position $m$ for $-L < m
< L$, and zeros elsewhere, and $\inprod{\ \cdot \ , \ \cdot \ }_F$ denotes the
standard inner product between two matrices, defined as $\inprod{\mA,\ \mB}_F
= \sum_{ij} \overline{a_{ij}} b_{ij} = \trace(\mA^\H \mB)$.

\begin{figure}
\centering
\includegraphics[width=3.5in]{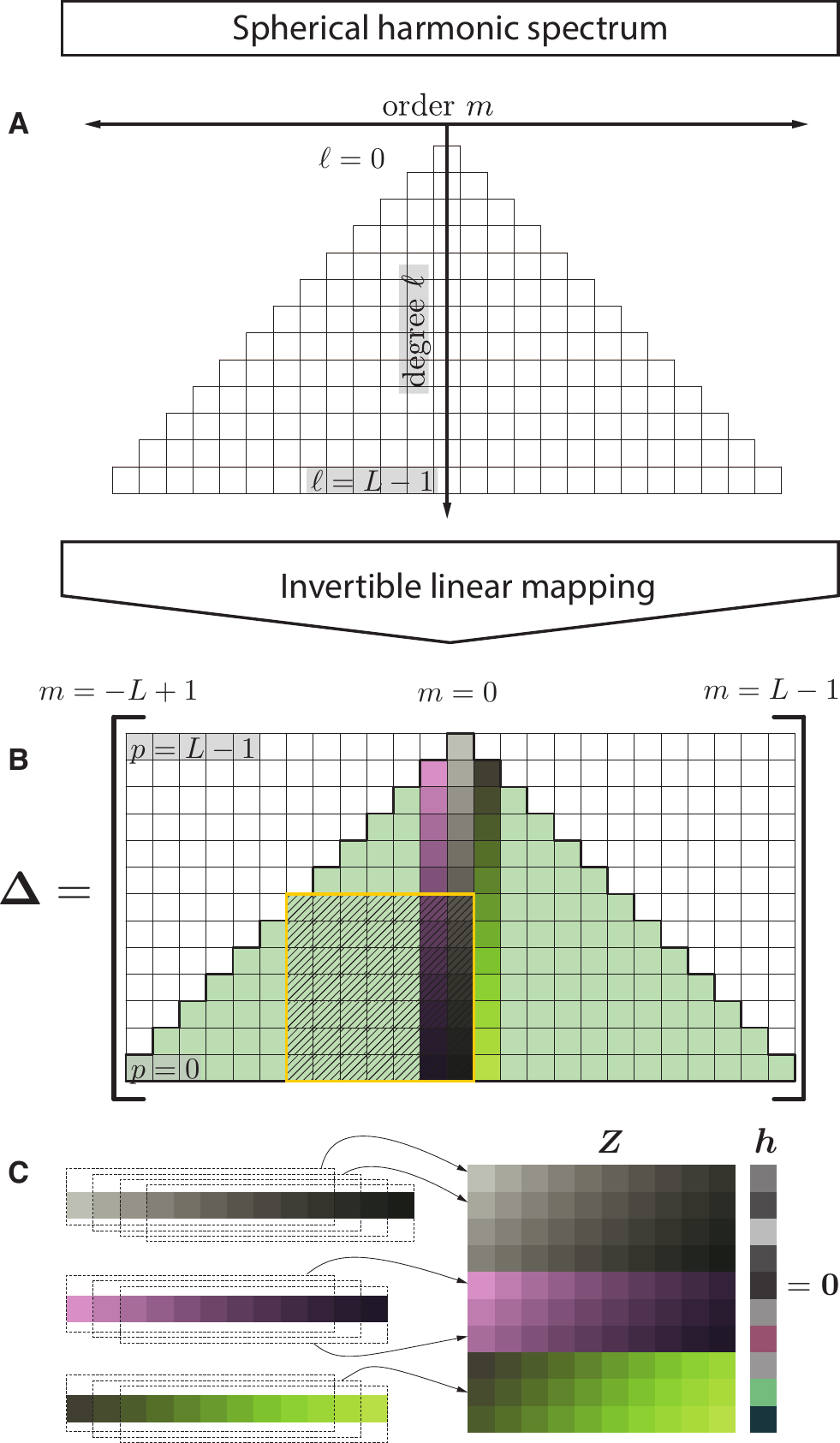}
\caption{Illustration of Algorithm \ref{alg:sfri}. Spherical harmonic spectrum
(A) is linearly mapped onto the shaded triangular part of the data matrix
$\mDelta$ (B). Columns of the data matrix are indexed from left to right by
$m$, $-(L-1) \leq m \leq (L-1)$, corresponding to spherical harmonic order.
Rows are indexed from bottom to top by $p$, $0 \leq p \leq (L-1)$
corresponding to powers of $\cos
 \theta$. Note that the triangular part of the data matrix \emph{does not}
coincide with the spherical harmonic spectrum, although there is a one-to-one
linear mapping between the two (see Lemma \ref{lem:fullrank}). Existing
results on 2D harmonic retrieval can exploit only a small part of the data
matrix, for example the hatched square (see Section
\ref{sub:relation_to_prior_work}. Finally, sufficiently long columns of
$\mDelta$ are rearranged in the block-Hankel-structured annihilation matrix
$\mZ$, whose nullspace contains exactly the sought annihilation filter, $\vh$
(C).}
\label{fig:data_matrix_illustration}
\end{figure}

The last expresion in \eqref{eq:linfunc} implies that the spectral coefficient
$\shf_{\ell}^m$ can be obtained as an inner product between the data matrix
$\mDelta$ and a mask $\vc_{\ell m} \ve_m^\T$ that is overlaid over $\mDelta$.
One can verify that the support of this mask for $\shf_\ell^m$ is on the
column corresponding to $m$, and on the rows corresponding to $0 \le p < L -
\abs{m}$. That means that certain parts of the data matrix are not involved in
the creation of any spectral coefficient; consequently, they cannot be
recovered from the spectrum. Nevertheless, we can recover a large part:

\begin{lem}
  \label{lem:fullrank}
  There is a one-to-one linear mapping between the spherical harmonic
  coefficients in the lowpass subband, $\big[ \shf_\ell^m, (\ell, m) \in
  \mathcal{I}_L \big]$, and the triangular part of the data matrix $\mDelta$
  indexed by $\mathcal{J}_L =
  \set{(p, m) \ | \ 0 \leq \abs{m} \leq p < L}$ (with indexing as illustrated
  in Fig.~\ref{fig:data_matrix_illustration}).
\end{lem}

\begin{proof}
  It is straightforward to verify that all the masks $\vc_{\ell m} \ve_m^\T$
  for $0 \leq \abs{m}
  \leq \ell < L$ are supported on the triangular part of $\mDelta$, as indexed
  by $\mathcal{J}_L$. Because the number of such masks coincides with the
  number of entries in the triangular part, and no mask is identically zero,
  it only remains to show that the masks are linearly independent. For $m_1
  \neq m_2$, this is true because their supports are disjoint ($\ve_{m_1}^\T$
  and $\ve_{m_2}^\T$ activate different columns),

  \begin{equation}
    \supp(\vc_{\ell_1 m_1} \ve_{m_1}^\T) \cap \supp(\vc_{\ell_2 m_2} \ve_{m_2}^\T) = \emptyset
  \end{equation}
  for any $\ell_1, \ell_2$. For $\ell_1 < \ell_2$ and $m_1 = m_2 = m$,
  $\vc_{\ell_1 m}^\T \vx$ and $\vc_{\ell_2 m}^\T \vx$ are polynomials of
  different degrees (\emph{c.f.} \eqref{eq:polyinprod}),
  \begin{equation}
    \deg(\vc_{\ell_1 m}^\T \vx) = (\ell_1 - \abs{m}) < (\ell_2 - \abs{m}) = \deg(\vc_{\ell_2 m}^\T \vx,
  \end{equation}
  where $\deg(\ \cdot \ )$ denotes the degree of the polynomial in the
  argument. Therefore, $\supp(\vc_{\ell_1 m} \ve_{m}^\T) \neq
  \supp(\vc_{\ell_2 m_2} \ve_{m}^\T)$, and in particular $\vc_{\ell_2 m}$ is
  linearly independent from all $\vc_{\ell m}$ such that $\ell < \ell_2$. This
  implies that all masks are linearly independent. Thus the mapping
  \begin{align}
    \label{eq:invertible_mapping}
     \mDelta 
     &\mapsto \big[ \inprod{\vc_{\ell m} \ve_m^\T,\ \mDelta}_F, \ 0 \leq \abs{m} \leq \ell < L \big] \nonumber \\
     &= \big[ \shf_\ell^m,  0 \leq \abs{m} \leq \ell < L \big]
  \end{align}
  is one-to-one on $\mathcal{J}_L$.
\end{proof}

\subsection{Reconstruction by Generalized Annihilating Filtering}

Element of the data matrix $\mDelta$ at the position $(p, m)$ (with reference
to Fig. \ref{fig:data_matrix_illustration}B) can be expanded as
\begin{equation}
  \label{eq:harmonic2d}
  d_{pm} = \sum_{k=1}^K \alpha_k x_k^{p} (\sin \theta_k)^{\abs{m}} e^{-\I m \phi_k},
\end{equation}
where $p$ varies from $0$ to $L-1$, and $m$ from $-(L-1)$ to $(L-1)$. For
either positive or negative $m$, the sum \eqref{eq:harmonic2d} is a sum of 2D
exponentials. Lemma~\ref{lem:fullrank} implies that we can recover the shaded
triangular part of the data matrix in Fig.~\ref{fig:data_matrix_illustration}
from the spectrum. In what follows, we propose a new algorithm to recover the
parameters of the Diracs from that triangular part.

The vector $\vd_m \bydef \mDelta \ve_m$ is a linear combination of columns of
$\mX$, \emph{i.e.}, it is a linear combination of $K$ exponentials with bases
$x_k$,
\begin{equation}
  \label{eq:d_m}
  d_{pm} = \sum_{k=1}^K (\alpha_k u_{km}) x_k^p,
\end{equation}
where $x_k = \cos(\theta_k)$. Similarly to standard Euclidean FRI sampling
\cite{Vetterli:2002bs}, we can use the
\emph{annihilating filter} technique to estimate the roots $\set{x_k = \cos
\theta_k}_{k=1}^K$ of these exponentials. 

Annihilating filter is a finite impulse response (FIR) filter with zeros
positioned so that it annihilates signals of the form
\eqref{eq:d_m}. Consider an FIR filter $H(z)$ with the transfer function
\begin{equation}
  \label{eq:annihilation_filter}
  H(z) \bydef \prod_{k=1}^K (1 - x_k z^{-1}) \bydef \sum_{n=0}^K h_n z^{-n},
\end{equation}
where $\vh = [h_0, \ h_{1}, \ \ldots, \ h_K]^\T$ is the vector of filter
coefficients. It holds that $\vh \conv \vd_m \equiv \vzero$ (see Appendix
\ref{appendix:annihilation_property}) for any $m$, provided that $\vd_m$ is of
length at least $K+1$. Equivalently, 
\begin{equation}
\label{eq:annihilation}
[d_{n, m}, \ d_{n-1,m}, \ \ldots, \ d_{n-K,m}] 
\begin{bmatrix}
h_0 \\ h_1 \\ \vdots \\ h_K
\end{bmatrix}
= 0,
\end{equation} 
for $n \geq K$. In our scenario, we do not know the bases of the exponentials
$\set{x_k}_{k=1}^K$---they are exactly the parameters we aim to estimate. Thus
we do not know the filter $H(z)$ either.
 
Up to a scaling factor, there is a unique $(K+1)$-tap filter $H(z)$ with the
sought property. The orthogonality relation \eqref{eq:annihilation} says that
$\vh$ lives in the nullspace of $[d_{n,m} \ d_{n-1,m} \ \cdots \ d_{n-K,m}]$;
we need at least $K$ such vectors to make their joint nullspace
one-dimensional, thus to pinpoint $\vh$. Once the filter coefficients are
found, we can obtain the unknown parameters $\set{x_k}$ by root finding and
using the factorization in
\eqref{eq:annihilation_filter}.

For the annihilating filter technique to be applicable, we need to ensure that
all the colatitude angles $\theta_k$ are distinct. Furthermore, the form of
our equations reveals that for $\theta_k \in \set{0, \pi}$, $u_{km} = 0$ for
all $m$. In the parameterization
\eqref{eq:data_matrix}, this is equivalent to setting $\alpha_k = 0$, and it
prevents us from recovering the corresponding Dirac. This behavior is
undesirable, but we can guarantee that no Dirac sits on a pole by first
applying a random rotation. This fact is formalized in the following lemma, which follows immediately from the absolute continuity of the Haar measure.


\begin{lem}
\label{lem:almost_surely_distinct}
Consider a collection of Dirac delta functions on the sphere, $f(\xi) =
\sum_{k=1}^K \alpha_k \delta(\xi; \xi_k)$, and a random rotation $\varrho$
drawn from the Haar measure on $\SO_3$ (\emph{i.e.} uniformly over the
elements of the group). Then with probability 1, $\Lambda(\varrho) f$ contains
Diracs with distinct colatitude angles, $\theta_i \neq \theta_j$ for $i \neq
j$, and no Dirac is on the pole, $\theta_k \notin \set{0, \pi}$ for all $k$.
\end{lem}
We are now well-equipped to prove the main result.

\begin{proof}[Proof of Theorem \ref{thm:sfri}]
  We provide a constructive proof, summarized in Algorithm \ref{alg:sfri}.
  First observe that $L^2$ random samples almost surely suffice to compute the
  spectral coefficients $\shf_{\ell}^m$ in the lowpass subband $\mathcal{I}_L$
  with bandwidth $L$, as detailed in Section
  \ref{sub:from_samples_to_spectrum} (see Proposition
  \ref{prop:random_sampling}). By Lemma \ref{lem:fullrank}, we can then
  compute the shaded part of $\mDelta$ given the spectrum $\wh{\vf}$.
  
  Our aim is to construct the annihilating matrix $\mZ$, structured as follows
  \begin{equation}
    \label{eq:annihilating_matrix}
    \mZ = 
    \begin{bmatrix}
      d_{L-1,0} & d_{L-2,0} & \cdots & d_{L-K-1, 0} \\
      d_{L-2,0} & d_{L-3,0} & \cdots & d_{L-K-2,0} \\
      \vdots & \vdots &  & \vdots \\
      d_{K, 0} & d_{K-1,0} & \cdots & d_{0,0} \\
      d_{L-2,1} & d_{L-3,1} & \cdots & d_{L-K-2,1} \\
      d_{L-3,1} & d_{L-4,1} & \cdots & d_{L-K-3,1} \\
      \vdots & \vdots & & \vdots
    \end{bmatrix}.
  \end{equation}
  $\mZ$ is constructed by stacking segments of length $(K+1)$ extracted from
  the columns of $\mDelta$. From the annihilation property
  \eqref{eq:annihilation}, it follows that the nullspace of $\mZ$ contains the
  sought annihilating filter.

  The trick now is to count how many such segments we can get from the shaded
  part of $\mDelta$. For $m=0$, $p$ varies from $0$ to $L-1$. Therefore, we
  can construct $L-K$ rows of the matrix $\mZ$. For $m=1$, $p$ varies from $0$
  to $L-2$, so we can construct $L-K-1$ rows of $\mZ$, and the same goes for
  $m=-1$. This process is illustrated in
  Figs.~\ref{fig:data_matrix_illustration}B and
  \ref{fig:data_matrix_illustration}C. Summing up, we get the total number of
  rows of $\mZ$ that we can construct from the available spectrum,
  \begin{equation}
  \begin{aligned}
    \# &= (L-K) + 2\times (L-K-1) + \cdots + 2 \times 1 \\
    &= (L-K)^2.
  \end{aligned}
  \end{equation}
  $\mZ$ needs at least $K$ rows, as we need a 1D nullspace. Thus
  \begin{equation}
  \begin{aligned}
    & \phantom{\Rightarrow} (L-K)^2 \geq K  \\
    &\Rightarrow L \geq K + \sqrt{K}.
  \end{aligned}
  \end{equation}
  In Appendix \ref{appendix:rank_of_Z} we show that $\mZ$ has rank $K$ as soon
  as it has $K$ or more rows. In other words, it has a one-dimensional
  nullspace, and thus the annihilating filter coefficients are uniquely
  determined, up to a scaling factor.

  We find the parameters $[\theta_k]_{k=1}^K$ by taking the arc cosine of the
  roots of $H(z)$. This procedure is well-posed because arc cosine is one-to-one on $[0,
  \pi]$. To ensure that the roots are distinct, we apply a random rotation
  before the estimation, and the inverse of this random rotation after
  recovering all the parameters of the Diracs (invoking
  Lemma~\ref{lem:almost_surely_distinct}).

  In order to recover the azimuths $\set{\phi_k}_{k=1}^K$, note that after recovering
  the colatitudes, we can construct the matrix $\mX$, and compute $\mA \mU
  \ve_m$ for $\abs{m} \leq L - K$. The azimuths are then given
  as the phase difference between $\mA \mU \ve_0$ and $\mA \mU \ve_1$. The
  magnitudes $\alpha_k$ are obtained simply as $\mA \mU \ve_0$.
\end{proof}

\begin{algorithm}[t]
\caption{Spherical Sparse Sampling}
\label{alg:sfri}
\begin{algorithmic}[1]
\Require{Spatial samples of $f \in L^2(\S^2)$ with bandwidth $L$, number of
Diracs $K$}
\Ensure{Colatitudes, azimuths and magnitudes $\set{(\alpha_k, \theta_k,
\phi_k)}_{k=1}^K$ of the $K$ Diracs}
\vspace{1mm}
\State Sample a random rotation $\varrho \sim \text{Haar}(\mathbb{SO}_3)$ 
\State Apply $\varrho$ to $f$, $f \gets \Lambda(\varrho) f$ (relabel sampling points)
\State Compute the spectrum $\shf$ from the rotated samples of $f$
\State Form $\mDelta$ from $\shf$ using the inverse mapping of \eqref{eq:invertible_mapping}
\State Form $\mZ$ from $\mDelta$ according to \eqref{eq:annihilating_matrix}
\State $\vh \gets$ Right singular vector of $\mZ$ for smallest sing. val.
\State Compute the colatitudes, $(\theta_k)_{k=1}^K \gets \arccos[\text{Roots}(\vh)]$
\State Construct $\mX$ from $x_k = \cos \theta_k$ according to \eqref{eq:matrix_X}
\State Using $\mX$ in \eqref{eq:data_matrix}, compute $\mA \mU \ve_0$ and $\mA
\mU \ve_1$
\State $(\phi_k)_{k=1}^K \gets \text{Angle} \big[(\mA \mU \ve_0) \oslash (\mA
\mU \ve_1) \big]$  \Comment{See the note}
\State $(\alpha_k)_{k=1}^K \gets \mA \mU \ve_0$
\State Apply the inverse of $\varrho$, $\xi_k = (\theta_k, \phi_k) \gets \varrho^{-1} \circ \xi_k$, $\forall k$
\end{algorithmic}
\hrule
\vspace{1mm}
$\rhd$ Note: we use the symbol $\oslash$ to denote element-wise division \\
\phantom{$\rhd$} of vectors.
\end{algorithm}

\subsection{Sampling Efficiency and Relation to Prior Work}

\label{sub:relation_to_prior_work}

Our proposed sampling scheme and the spherical FRI sampling theorem by
Deslauriers-Gauthier and Marziliano \cite{DeslauriersGauthier:2013ij} are both
naturally expressed in terms of the bandwidth $L$ of the sampling kernel
required to recover $K$ Diracs. In our case, the bandwidth requirement is that
it be at least $K + \sqrt{K}$. This implies that we need at least $(K +
\sqrt{K})^2$ spatial samples in order to recover the $K$ Diracs. For
comparison, the FRI sampling theorem of Deslauriers-Gauthier and Marziliano
\cite{DeslauriersGauthier:2013ij} requires $L \ge 2K$, and thus their
algorithm recovers $K$ Diracs given $4K^2$ samples. This is asymptotically
four times the number of samples required by Algorithm
\ref{alg:sfri}. 

The difference in sampling efficiency can be explained by spectrum usage.
Fig.~\ref{fig:spectrum_efficiency} illustrates the portion of the spectrum
used by the two algorithms. We can see that the proposed algorithm is more
efficient in that it uses a larger portion of the available spectrum to
reconstruct the Diracs.

Similar problems have been considered in the literature on 2D harmonic
retrieval \cite{Vanpoucke:1994bq}. However, these earlier works assume that
the entire data matrix is known. In our case,
$\mDelta$ is known only partially, as illustrated in Fig.
\ref{fig:data_matrix_illustration}B.
To apply the existing results on 2D harmonic retrieval, we could use a square
portion that falls strictly inside a half of the triangle, either for $m
\ge 0$ or for $m \le 0$. However, we can see in
Fig.~\ref{fig:data_matrix_illustration}B that this is an inefficient use of
available spectrum, and it requires an unnecessarily high sampling
density.

As mentioned earlier, in most situations we do not get to choose $L$ as
it is fixed by the underlying physical process. Then the question is how many
Diracs we can reconstruct given a kernel with a fixed bandwidth $L$. By
solving $L \geq K + \sqrt{K}$ for $K$, we get that
\begin{align}
  K \leq  L - (L + \tfrac{1}{4})^{1/2} + \tfrac{1}{2}.
\end{align}
In contrast, the algorithm in \cite{DeslauriersGauthier:2013ij} can
reconstruct up to $K=L/2$ Diracs.

\begin{figure}
  \centering
  \begin{overpic}[width=3.5in]{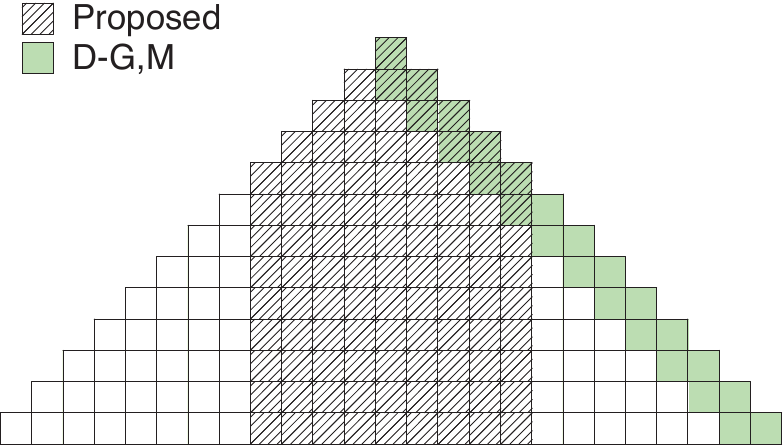}
  \put (23.4, 48.20) {\scalebox{1.15}{\textsf{\cite{DeslauriersGauthier:2013ij}}}}
  \end{overpic}
  \caption{Spectrum usage for different algorithms. Spectral coefficients used
  by our algorithms are shown hatched. Spectrum used by the algorithm of
  Deslauriers-Gauthier and Marziliano \cite{DeslauriersGauthier:2013ij} is
  shaded green. In the example, the bandwidth is set to $L=12$, so the maximum
  number of Diracs that can be recovered by Algorithm \ref{alg:sfri} is $K=9$.
  The algorithm in \cite{DeslauriersGauthier:2013ij} recovers $K=6$ Diracs.}
  \label{fig:spectrum_efficiency}
\end{figure}

\subsection{Denoising Strategies} 
\label{sub:denoising_str}

Theorem \ref{thm:sfri} and Algorithm \ref{alg:sfri} provide a tool to recover
sparse signals on the sphere in the noiseless case. We may apply several
procedures to improve the robustness of the algorithm in the presence of
noise.

In general, if the samples are noisy then the annihilating matrix $\mZ$ in
\eqref{eq:annihilating_matrix} will not have a nontrivial nullspace. A simple
and robust approach is to use the right singular vector corresponding to the
smallest singular value of $\mZ$ as the annihilation filter. Let $\mZ = \mU
\mat{\Sigma} \mV^H$ be the SVD of $\mZ$; then we set $\vh =
\vv_{K+1}$.

To further improve the algorithm performance, we can use the output of
Algorithm \ref{alg:sfri} to initialize a local search for the minimizer of the
$\ell^2$ error between the spectrum generated by the estimated Diracs, and the
measured spectrum,
\begin{equation}
  \label{eq:descent}
  \mathop{\text{minimize}}_{(\wt{\alpha}_k, \wt{\theta}_k, \wt{\phi}_k)_{k=1}^K}~\sum_{\ell=0}^{L-1} \sum_{m=-\ell}^\ell \abs{\shf_\ell^m - \sum_{k=1}^K \wt{\alpha}_k Y_\ell^m(\wt{\theta}_k, \wt{\phi}_k)}^2.
\end{equation}
We note that directly solving \eqref{eq:descent} with a random starting point
is hopeless due to a multitude of local minima.

\subsection{Cram\'er-Rao Lower Bound} 
\label{sub:crlb}

We evaluate the Cram\'er-Rao lower bound (CRLB) for the estimation problem.
For simplicity we treat the $K=1$ case, so that the minimal bandwidth is $L =
2$, and $\ell \in \set{0, 1}$. We assume that the spatial samples are taken on
the sampling grid defined by McEwen--Wiaux \cite{McEwen:2011ib}, given at this
bandwidth as
\begin{equation}
  \begin{bmatrix}
    \bm{\theta} \\
    \bm{\phi}
  \end{bmatrix}
  = 
  \begin{bmatrix}
    \pi/3 & \pi/3 & \pi/3  & \pi    & \pi    & \pi \\
    0     & 0     & 2\pi/3 & 2\pi/3 & 4\pi/3 & 4\pi/3
  \end{bmatrix}.
\end{equation}
Resulting expressions for elements of the Fisher information matrix are
complicated, and there is no need to exhibit them explicitly. We give the
details of the computation in Appendix \ref{appendix:crlb_computation}, and we
compute the CRLB numerically. The resulting bound is plotted in
Fig.~\ref{fig:crlb} for two different spike colatitudes, together with the MSE
achieved by Algorithm \ref{alg:sfri} followed by the descent
\eqref{eq:descent}. As pointed out before, because our scheme is
coordinate-system-dependent, the bound depends on the colatitude of the Dirac.

\begin{figure}
\centering
\includegraphics{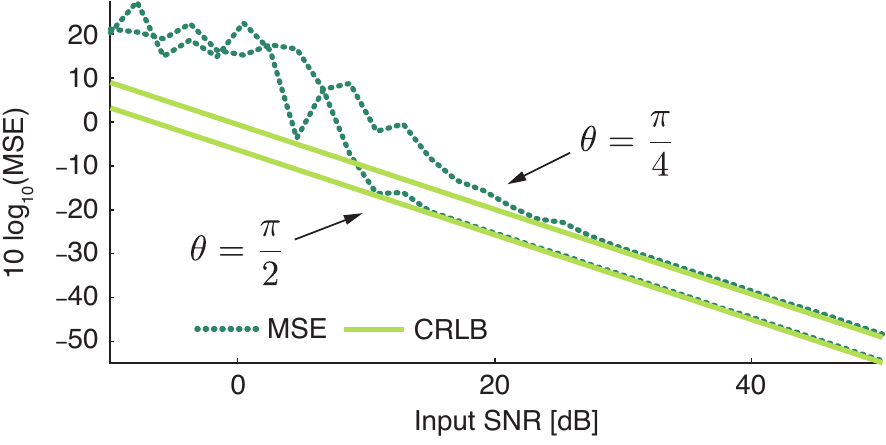}
\caption{Comparison between the mean squared error (MSE) of the proposed
algorithm in estimating the spherical location $(\theta, \phi)$, with $K=1$,
and the Cram\'er-Rao lower bound (CRLB), at two different colatitudes. Note
that the bound is different for different colatitudes of the spike, due to
parameterization dependence. MSE is shown for the output of Algorithm
\ref{alg:sfri} followed by the minimization of \eqref{eq:descent} using
Matlab's \texttt{fminsearch} function.}
\label{fig:crlb}
\end{figure}

\section{Applications}

\label{sec:applications}

To showcase the versatility of the proposed algorithm, we present three
stylized applications: 1) sampling diffusion processes on the sphere, 2)

shot noise removal, and 3) sound source localization with spherical
microphone arrays.

\subsection{Sampling Diffusion Processes on the Sphere}

\begin{figure}[t!]
  \centering
  \includegraphics{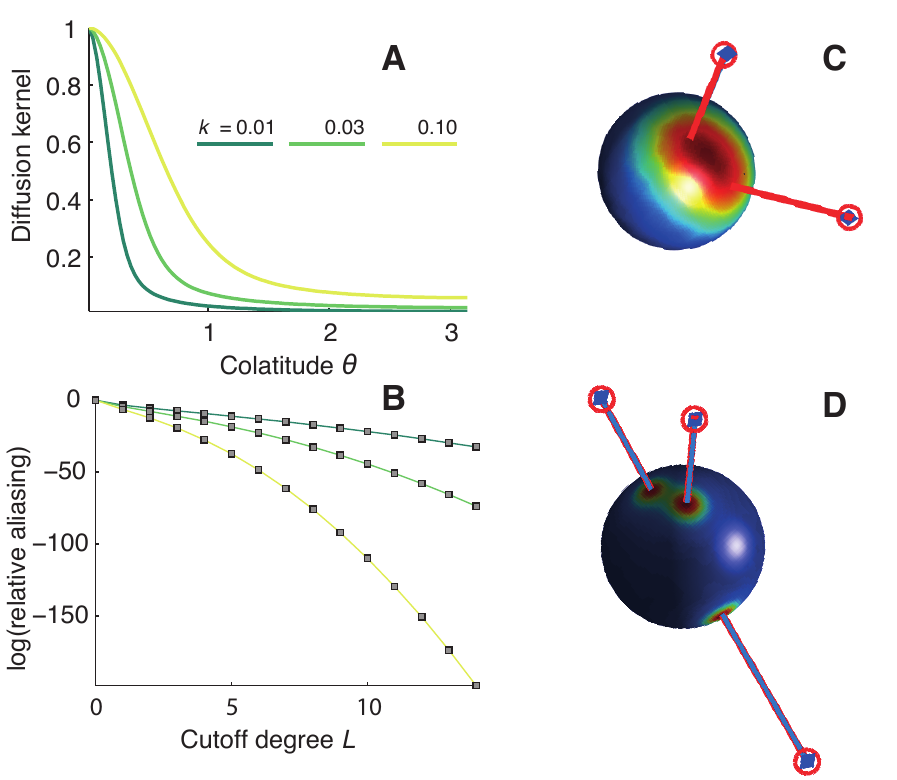}
  \caption{Estimating the release locations and magnitudes of diffusive
  sources on the sphere. We assume that the diffusive sources appear at time
  $t=0$s, and that we sample the field at time $t = 1$s. Shape of the
  diffusion kernel as a function of $\theta$ is shown in subfigure A for three
  different values of the coefficient $k$ (in units of inverse time). The
  logarithm of the aliasing error \eqref{eq:aliasing} is plotted as a function
  of the cutoff degree $L$ in subfigure B. Subfigures C and D show a typical
  reconstruction result for $k=0.1$ (2 sources) and $k=0.01$ (3 sources).
  Magnitudes of the sources are represented by the distance of the
  corresponding symbols from the sphere's center. Blue diamonds represent true
  source locations and magnitudes, while red circles represent estimated
  source locations and magnitudes. The sphere color corresponds to the value
  of the function induced on the sphere by the sources (red is large, blue is
  small). Signal-to-noise ratio in both C and D was set to 30 dB. We used the
  approximate bandwidth of $L=7$, so that the number of samples taken in
  either case was 49.}
  \label{fig:Diffusion}
\end{figure}

The diffusion process models many natural phenomena. Examples include heat
diffusion and plume spreading from a smokestack. Often, the source of the
diffusion process is localized in space, and instantaneous in time. Sampling
such processes in Euclidean domains has been well studied
\cite{Dokmanic:2011iy, Lu:2009go, Lu:2011id}.

Diffusion processes on the sphere are governed by the equation
\cite{Bulow:2004jv}

\begin{equation}
  k \Delta v(\xi, t) = \parder{}{t} v(\xi, t),
\end{equation}
where $\Delta$ is the Laplace-Beltrami operator on $\S^2$, and $k$ is the
diffusion constant. In the spherical harmonic domain, this becomes
\begin{equation}
 -k \ell (\ell+1) \shv_{\ell}^m(t) = \parder{}{t} \shv_{\ell}^m(t),
\end{equation}
giving the solution
\begin{equation}
  \label{eq:diffusion_initial}
  \shv_{\ell}^m(t) = e^{-\ell(\ell+1) k t} \shv_{\ell}^m(0),
\end{equation}
where $\shv_{\ell}^m(0)$ is the spectrum of the initial distribution.
Therefore, we interpret the term $e^{-\ell(\ell+1)kt} \delta_{m, 0}$ as the
spectrum of the Green's function of the spherical diffusion equation. In other
words, it is the spectrum of the diffusion kernel on the sphere. Then
\eqref{eq:diffusion_initial} should be interpreted as the convolution between
the kernel and the initial distribution.

We consider the case when the diffusion process is initiated by $K$ sources
localized in space and time, \emph{i.e.}, the initial distribution in
\eqref{eq:diffusion_initial} is
\begin{equation}
  v(\xi; 0) = \sum_{k=1}^K \alpha_k \delta(\xi; \xi_k).
\end{equation}
We show how to use the proposed sampling algorithm to estimate the locations
and the strengths of the sources from spatial samples of the diffusion field
taken at a later time $t_0$. 
Recovering all parameters (locations, amplitudes and release times) of
multiple diffusion sources is a challenging task \cite{Dokmanic:2011iy}. To
focus on the proposed sampling result, we make the simplifying assumption that
the $K$ sources are released simultaneously, and at a known time ($t=0$). In
principle, the more challenging case of unknown and different release times
can be handled by adapting the techniques derived in \cite{Lu:2009go,
Lu:2011id}, but these generalizations are out of the scope of this work.

In the spatial domain, the diffusion kernel at time
$t_0$ after the release is given as
\begin{equation}
  \label{eq:diffusion_kernel_spatial}
  h_{\mathrm{dif}}(\xi; t_0) = \sum_{\ell=0}^{\infty} e^{-\ell(\ell+1)k t_0} Y_\ell^0(\xi).
\end{equation}
Combining \eqref{eq:diffusion_kernel_spatial} with \eqref{eq:diffusion_initial} and the spherical convolution-multiplication rule \eqref{eq:ConvSpectrum}, we get
\begin{equation}
\begin{aligned}
  v(\xi; t_0) &= v(\xi; 0) \conv h_{\mathrm{dif}}(\xi; t_0) \\
  &= \sum_{k=1}^K c_k [\Lambda(\varrho_k) h_{\mathrm{dif}}(\ \cdot \, ; t_0)](\xi).
\end{aligned}
\end{equation}
This signal is a sum of rotations of a known template. The diffusion kernel in
\eqref{eq:diffusion_kernel_spatial} is not exactly bandlimited, but it is
approximately so. We can therefore apply the spherical FRI theory and
Algorithm
\ref{alg:sfri} to recover the locations and the magnitudes of the diffusive
sources.

Fig.~\ref{fig:Diffusion}A shows the shape of the symmetric diffusion kernel as
a function of the colatitude $\theta$. The high degree of smoothness is
reflected in an approximately bandlimited spectrum. This is demonstrated in
Fig.~\ref{fig:Diffusion}B, where we see that the aliasing energy due to
spectral truncation, defined as

\begin{equation}
  \varepsilon(L) = \frac{1}{\norm{v}_2^2}\sum_{\ell=L}^\infty \frac{\abs{\shv_\ell^m}^2}{2\ell + 1},
  \label{eq:aliasing}
\end{equation}
rapidly becomes negligible as we increase the cutoff bandwidth $L$.
Figs.~\ref{fig:Diffusion}C and
\ref{fig:Diffusion}D demonstrate accurate reconstruction of the localized
diffusion sources at two different values of the diffusion coefficient (the
detailed parameters of the numerical experiment are given in the figure
caption).

\subsection{Shot Noise Removal}

Suppose that we sample a bandlimited function on the sphere, but a small
number of samples are corrupted---they contain \emph{shot noise}---due to
sensor malfunction. Moreover, the identities of the malfunctioning sensors are
not known \emph{a priori}. Can we detect and correct these anomalous
measurements? We show that our sampling results can be applied to solve this
problem, provided that the number of erroneous sensors is not too large and
that the original sampling grid is \emph{oversampling} the bandlimited
function. A similar idea was used in \cite{Marziliano:2006do} to remove shot
noise in the 1D Euclidean case.

For this application we assume that the samples are taken on a uniform grid on
the sphere, $\set{(\theta_p, \phi_q) \ | \ p, q \in \Z, 0 \leq p < 2L', 0 \leq q < 2L'}$,
defined by
\begin{equation}
  \label{eq:dh_grid}
  \theta_p = \dfrac{p \pi}{2 L'} \quad , \quad \phi_q = \dfrac{q \pi}{L'}.
\end{equation}

Imagine now that we sample $f$ on this sampling grid. Some samples are
corrupted, so we measure $g(\theta_p, \phi_q) = f(\theta_p,
\phi_q) + s_{pq}$, where
\begin{equation}
  s_{pq} = 
  \begin{cases}
    \text{nonzero} & (p,q) \in \mathcal{S} \\
    \text{zero} & \text{otherwise},
  \end{cases}
\end{equation}
and $\mathcal{S}$ holds the indices of the corrupted samples.  We will
leverage an elegant quadrature rule by Driscoll and Healy
\cite{Driscoll:1994bp}:
\begin{thm}{\cite[Theorem 3]{Driscoll:1994bp}}
  Let $f$ be a bandlimited function on $\S^2$ such that $\shf_\ell^m =
  0$ for $\ell \geq L'$. Then for $(\ell, m) \in \mathcal{I}_{L'}$ we have
  \begin{equation}
    \shf_\ell^m = \sum_{p=0}^{2L'-1} \sum_{q=0}^{2L'-1} a_p^{(L')} f(\theta_p, \phi_k) \overline{Y_\ell^m(\theta_p, \phi_q)},
  \end{equation}
  where the weights $a_p^{(L')}$ are defined in \cite{Driscoll:1994bp}.
\end{thm}
In other words, the Fourier coefficients $\shf_\ell^m$ can be expressed as a
dot-product between weighted sample values and the basis functions evaluated
at the sampling points. In analogy with the Euclidean case, we now observe
that the lowpass portion of the spectrum of $f$ coincides with the lowpass
portion of the spectrum of the generalized function obtained by placing
weighted Diracs at grid points. Let $f$ be bandlimited so that $\shf_\ell^m =
0$ for $\ell
\geq L$. Let further $L < L'$; that is, the grid \eqref{eq:dh_grid}
oversamples $f$. Then the spectral coefficients can be expressed as the
following inner product,
\begin{equation}
  \shf_\ell^m = \inprod{\sum_{p, q=0}^{2L'-1} a_p^{(L')} f(\theta_p, \phi_q) \delta_{\theta_j, \phi_q},\ Y_\ell^m},
\end{equation}
for $\ell < L$, $\abs{m} \leq \ell$.

This is the key insight. Notice that the
lowpass portion of the spectrum of $g$ (for $\ell < L'$) can be written as
\begin{equation}
  \shg_\ell^m = \shf_\ell^m + \inprod{\sum_{(p, q) \in \mathcal{S}} a^{(L')}_p s_{pq} \delta_{\theta_p, \phi_q}, Y_\ell^m}.
\end{equation}
But $\shf_\ell^m = 0$ for $\ell \geq L$, so the portion of the spectrum for $L
\leq \ell < L'$ contains only the influence of the corrupted samples, 
\begin{equation}
  \shg_\ell^{m} = \inprod{\sum_{(p,q) \in \mathcal{S}} a_p^{(L')} s_{pq} \delta_{\theta_p, \phi_q}, Y_\ell^m}, \quad L \leq \ell < L'.
\end{equation}
Consequently, we can use this part of the spectrum to learn which samples are
corrupted, and by how much. This is the subject of the following proposition.

\begin{prop}
  Let $f$ be a signal on the sphere of bandwidth $L$. Then we can perfectly
  reconstruct $f$ from corrupted samples taken on the grid \eqref{eq:dh_grid},
  as long as the number of corruptions $K$ satisfies
  \begin{equation}
    K \leq L' - L - \sqrt{L' - L + 1} + 1.
  \end{equation}
\end{prop}

\begin{proof}
As discussed in Section \ref{sec:spherefri}, we can use any $m =
\text{const.}$ line in the spectrum to get the rows of the annihilation
matrix. However, we first need to compute the corresponding columns of the
data matrix. From  Fig.~\ref{fig:shot_noise_spectrum}, we see that the middle
columns cannot be used for shot noise removal: we seek columns influenced only
by corruptions. But the middle columns of the data matrix are obtained from
the middle spectral columns (for $m < L$), so they are influenced both by the
desired signal and the corruptions. This means that we can only use spectral
bins for $L \leq m < L'$, as illustrated in
Fig.~\ref{fig:shot_noise_spectrum}. For $m = L$ and $m = -L$, the number of
segments of length $K+1$ that we can get is $L'-L-K$. For $m = L+1$ and $m =
-(L+1)$ it is $L'-L-K-1$, and so on. Summing up we have that the total number
of consecutive segments of length $K+1$ we can use is
\begin{align*}
  \# &= 2(L'-L-K) + 2(L'-L-K-1) + \cdots + 2 \cdot 1 \\
  &= (L'-L-K)(L'-L-K+1).
\end{align*}
We need this number to be at least $K$, because we need $K$ rows in the
annihilation matrix. We thus obtain the claim of the proposition by solving
the inequality $\# \geq K$.
\end{proof}

\begin{figure}
\centering
\includegraphics[width=3.5in]{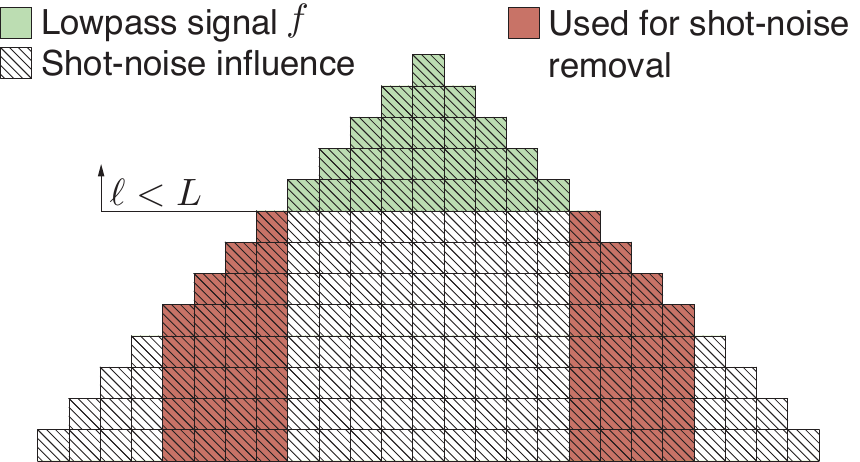}
\caption{Spectrum structure in shot noise removal. Green-shaded
bins get contribution from the desired signal $f$ with bandwidth $L$; hatched
bins are influenced by the shot noise; red-shaded columns are (i) long enough
to annihilate shot noise and (ii) recoverable from the corrupted spectrum.}
\label{fig:shot_noise_spectrum}
\end{figure}

\begin{figure}
  \centering
  \includegraphics[width=3.5in]{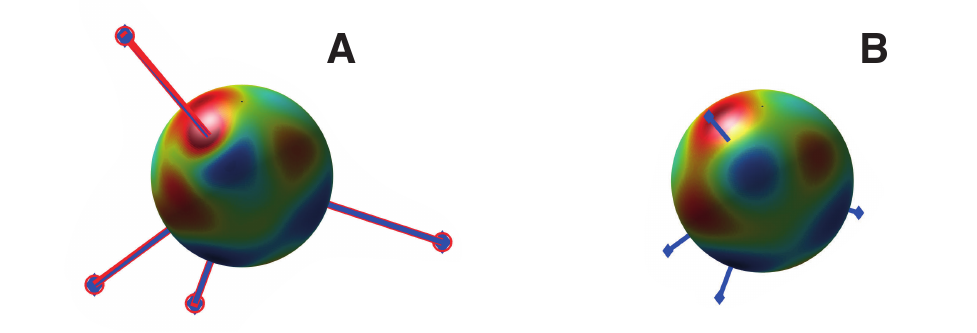}
  \caption{Shot noise removal via spherical FRI, for $L=6$, $L'=12$ and $K=4$
  malfunctioning sensors. Corrupted signal is shown in subfigure A, together
  with the true corruption values (blue diamonds) and the estimated
  corruptions (red circles); same signal with the shot noise removed is shown
  in B, with the correct sample values at the corrupted locations denoted by
  blue diamonds.}
  \label{fig:shot_noise_removal}
\end{figure}

\begin{figure*}[t!]
  \centering
  \includegraphics{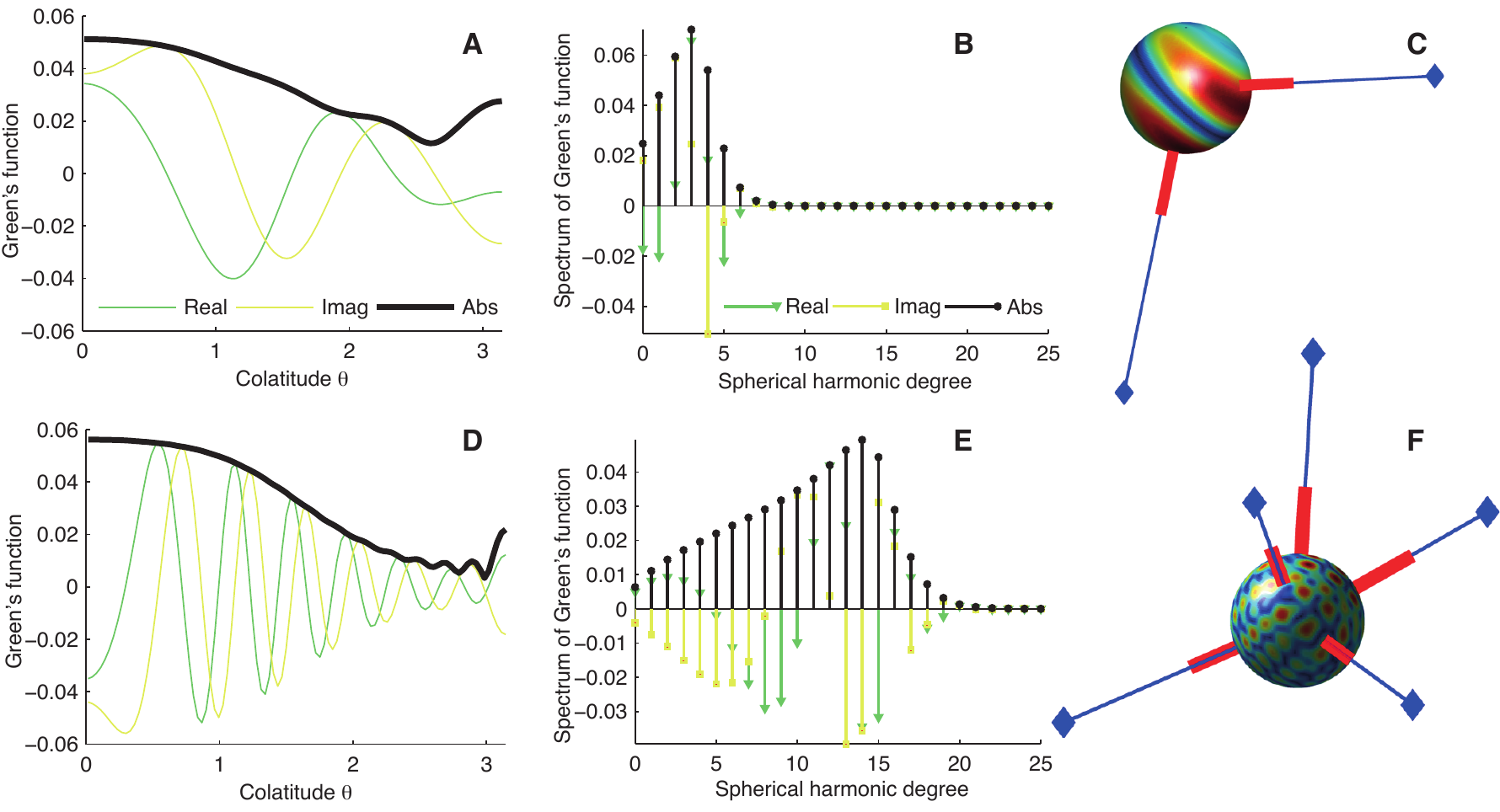}
  \caption{Multiple DOA estimation by a spherical microphone array. First row of
    subfigures corresponds to $f_1 = 1000$ Hz, and second row to $f_2 = 4000$
    Hz. Sphere has a radius $r = 0.2$ m, and the source is located at $[0,\
    0,\ 3]^\T$ m. The real and imaginary parts, and the absolute value of the
    Green's function are shown in subfigures A and D. Real part, imaginary
    part and absolute value of the spectrum are shown in subfigures B and E.
    Subfigures C and F show the simulation results for $K=2$ and $K=5$, and
    random source placement. Blue diamonds represent the source locations, and
    thick red lines show the estimated directions. Size of the sphere is
    exaggerated for the purpose of illustration. The sphere color corresponds
    to the absolute value of the function induced on the sphere by the sources
    (microphones measure samples of this function). The bandwidth was set to
    $L = 12$ at 1000 Hz and to $L = 30$ at 4000 Hz.}
  \label{fig:SSL}
\end{figure*}

After detecting the corrupted readings, we can use the estimated corruption
values to estimate the function. Another option is to simply ignore them
altogether, as we have more samples than the minimum number thanks to
oversampling. A shot noise removal experiment is illustrated in
Fig.~\ref{fig:shot_noise_removal}.

\subsection{Sound Source Localization}
\label{sub:ssl}

Spherical microphone arrays output a time-varying spherical signal. If the
signal is induced by a collection of point sources, we can use the proposed
spherical FRI sampling scheme to estimate the directions-of-arrival (DOAs) of
the sources. For simplicity, we consider the narrowband case, \emph{i.e.}, the
sources emit a single sinusoid.

How does this example fit into our sparse sampling framework? In spherical microphone
arrays, the microphones are distributed on the surface of a sphere, either
open or rigid \cite{Jarrett:2012fi}. Therefore, the microphone signals
represent samples of a time-varying function on $\S^2$. If a sound source
emits a sinusoid, every microphone measures the amplitude and the phase of
that sinusoid shaped by the characteristics of the propagating medium and of
the spherical casing. Equivalently, for every microphone we get a complex
number.

Suppose that a source of unit intensity is located at $\vs$, and that the
microphones are mounted on a rigid sphere of radius $r$ with center at the
origin. The response measured by the microphone at $\vr$, such that $\norm{\vr}
= r$, is given by the corresponding Green's function. For a wavenumber $\kappa = 2\pi
\nu / c$, where $\nu$ is the frequency and $c$ is the speed of sound, the Green's
function is \cite{Jarrett:2012fi}
\begin{equation}
  \label{eq:green}
  g(\vr|\vs, \kappa) = \frac{\I k}{4\pi} \sum_{\ell=0}^{\infty} b_\ell(\kappa r) h_\ell^{(1)}(\kappa s)(2\ell+1) P_{\ell}(\cos \alpha_{\vr \vs}),
\end{equation}
where $h_{\ell}^{(1)}$ is the spherical Hankel function of the first kind and
of order $\ell$, $P_\ell$ is the Legendre polynomial, and $\cos \alpha_{\vr
\vs}  = \frac{1}{rs}\inprod{\vr, \vs}$. Mode strength $b_{\ell}(kr)$ is
defined as
\begin{equation}
  b_{\ell}(\kappa r) \bydef j_{\ell}(\kappa r) - \frac{j_{\ell}^{\prime}(\kappa r)}{h_{\ell}^{(1) \prime}(\kappa r)}h_{\ell}^{(1)}(\kappa r),
\end{equation}
where $j_{\ell}$ is the spherical Bessel function\footnote{We use the standard
symbol $j_\ell$ for the spherical Bessel function. Note the subtle difference
from the imaginary unit $\I$.} of order $\ell$, and prime $(\,
\cdot \,)^{\prime}$ denotes the derivative with respect to the argument.

The Green's function $g$ should be seen as a filter that describes how the
point source's influence spreads over the sphere. It is shown for two
different frequencies in Figs. \ref{fig:SSL}A and \ref{fig:SSL}D, while the
corresponding spectra are given in Figs. \ref{fig:SSL}B and
\ref{fig:SSL}E. We see that the absolute pressure on the sphere has a
similar shape for both frequencies, but the real and imaginary parts vary
faster at higher frequencies, implying higher bandwidth. In both cases we
observe that the Green's function is approximately bandlimited.

Assume now that there are $K$ sound sources at locations $\set{\vs_k}_{k=1}^K$,
with complex intensities $\set{\alpha_k}_{k=1}^K$. The resulting measurement by
a microphone at point $\vr$ is
\begin{equation}
  \label{eq:K-sources}
  f(\vr) = \sum_{k=1}^K \alpha_i g(\vr|\vs_k, \kappa). 
\end{equation}

If all the source locations $\vs_k$ are at the same distance from the sphere,
then the Green's function \eqref{eq:green} depends only on the angle between
$\vr$ and $\vs$. For some fixed source distance $d$, we can define
$h_\mathrm{SSL}(\xi) \bydef g(\vx_\xi r| \vx_\eta d, \kappa)$, where
$\vx_\xi$ denotes the unit vector corresponding to $\xi$, $\vx_\eta$ the
unit vector corresponding to the north pole $\eta$, and the subscript SSL
stands for \emph{sound source localization}. Then
\eqref{eq:K-sources} corresponds to a weighted sum of $K$ rotations of a known
template function $h_\mathrm{SSL}$,
\begin{equation}
  f(\xi) = \sum_{k=1}^K \alpha_k h_\mathrm{SSL}(\varrho_k^{-1} \circ \xi).
\end{equation}

\begin{figure}[t]
  \centering
  \includegraphics{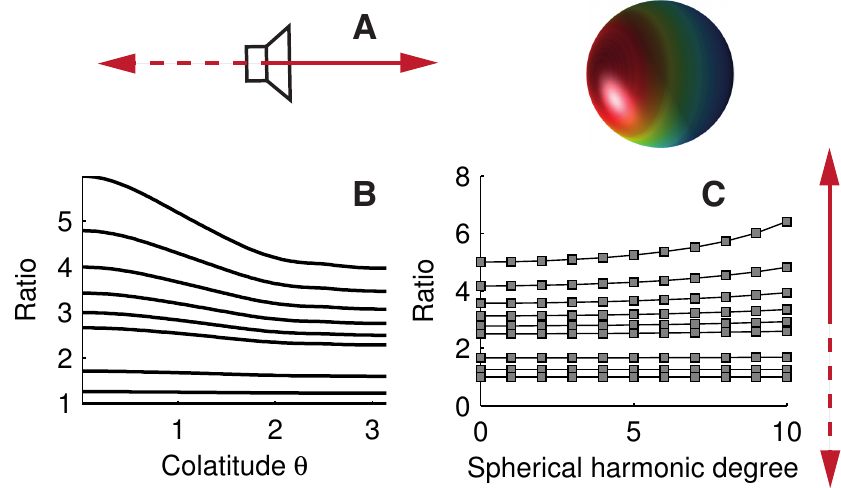}
  \caption{Ratios of Green's functions. We computed the Green's function for nine different source distances (1.0 m, 1.2 m, 1.4 m, 1.6 m, 1.8 m, 2.0 m, 3.0 m, 4.0 m, 5.0 m). Then we plotted the magnitude of the ratio of the Green's function at each distance and the Green's function at the largest distance (5 m), both in space (B) and in the spectrum (C). The more parallel the ratio curve is with the abscissa axis, the more similar the Green's function at that distance is to the Green's function at 5 m. Curves are plotted in the order of increasing distance in the direction of the dashed arrow (up to down), as indicated in (A).}
  \label{fig:ratios}
\end{figure}

As it is unrealistic to assume that the sources are all at the same distance,
we hope that the shape of $g(\vr | \vs, \kappa)$ does not (strongly) depend on
$\norm{\vs}$. Indeed, it turns out that the shape is approximately preserved
within a certain range, as illustrated in Fig.~\ref{fig:ratios}. We therefore
suppress the dependency of $g$ on $\norm{\vs}$ and approximate
\eqref{eq:K-sources} as follows,

\begin{equation}
  \begin{aligned}
    f(\xi)
    &= \sum_{k=1}^K \alpha_k g(\xi | \vs_k, \kappa) \\
    &\approx \sum_{k=1}^K \wt{\alpha}_k h_\mathrm{SSL}(\varrho_k^{-1} \circ \xi) \\
    &= \left[ \sum_{k=1}^K \wt{\alpha}_k \delta(\xi; \xi_{\vs_k}) \right]
    \conv h_\mathrm{SSL}(\xi).
  \end{aligned}
\end{equation}
Here, we absorbed $\alpha_k$ and additional (complex) scaling due to different
distances into $\wt{\alpha}_k$, and $h_\mathrm{SSL}$ is computed at some
predefined {\em mean} distance.

We thus reduced the sound source localization problem to a problem of finding
the parameters of a weighted sum of Diracs. In order to apply our spherical FRI
algorithm, we need to verify that $g$ is bandlimited on the sphere. Figs.
\ref{fig:SSL}D and \ref{fig:SSL}E show that it is indeed approximately
bandlimited, and that the bandwidth depends on the frequency (it also depends
on the sphere radius).

Figs. \ref{fig:SSL}C and \ref{fig:SSL}D show an example of recovering two
sources at 1000 Hz and 5 sources at 4000 Hz using the proposed spherical
sparse sampling scheme. It is worth noting that this succeeds in spite of the
model mismatch due to varying source distances. This indicates the robustness
of the proposed reconstruction algorithm.

\section{Conclusion}
\label{sec:conclusion}

We presented a new sampling theorem for sparse signals on the sphere. In
particular, by leveraging ideas from finite rate-of-innovation sampling, we
showed how to reconstruct sparse collections of spikes on the sphere from their
lowpass-filtered observations. Compared to existing sparse sampling schemes on
the sphere, we use the available spectrum more efficiently by generalizing
known results on 2D harmonic retrieval, thereby reducing the number of samples
required to reconstruct the parameters of the spikes.
 
We illustrated the usefulness of our algorithm by using it to solve three
problems: sampling diffusion processes, shot noise removal, and sound source
localization. But there is a wealth of other applications, for example in
astronomy. Just think about the numerous spherical signal processing
challenges put forward by the square kilometer array (SKA) project
\cite{Dewdney:2009bq}.

We mentioned some approaches to estimation from noisy samples, but more
efficient denoising schemes should be studied. One example, effective in the
Euclidean setting, is the Cadzow denoising algorithm \cite{Cadzow:1988gp}. The
problem seems more challenging on the sphere; in particular, the annihilating
matrix is block-Hankel, rather than Hankel.


\appendix
\section*{}


\subsection{Annihilating Property}

\label{appendix:annihilation_property}

For the sake of completeness, we show in this appendix that the annihilation
filter annihilates linear combinations of exponentials. We compute the
response of the filter $H(z)$ in \eqref{eq:annihilation_filter} to a signal of
the form $y_n = \sum_{k=1}^K b_k x_k^n$ as
\begin{align*}
  \label{eq:annihilation_explanation}
  (y \conv h)_n
  &= \sum_{m = 0}^K y_{n-m} h_m  
  = \sum_{m = 0}^K \left( \sum_{k = 1}^K b_k x_k^{n-m} \right) h_m \nonumber \\
  &= \sum_{k=1}^K x_k^n b_k \sum_{m=0}^K h_{m}x_k^{-m} \nonumber 
  = \sum_{k=1}^K x_k^n b_k \prod_{i=1}^K (1 - x_k x_i^{-1}) \nonumber \\
  &= 0.
\end{align*}


\subsection{Computation of the Cram\'er-Rao Lower Bound}
\label{appendix:crlb_computation}

A lowpassed collection of $K$ Diracs can be written as follows,

\begin{equation}
  f(\theta, \phi) = \sum_{\ell = 0}^{L-1} \sum_{m=-\ell}^{\ell} \left( \sum_{k=1}^K \alpha_k \overline{Y_{\ell}^m(\theta_k, \phi_k)} \right) Y_{\ell}^m(\theta, \phi).
\end{equation}
In the remainder of this section, we assume $K=1$, so we rewrite the function as
\begin{equation}
  f(\theta, \phi) = \sum_{\ell = 0}^{L-1} \sum_{m=-\ell}^{\ell} \alpha_0 \overline{Y_{\ell}^m(\theta_0, \phi_0)}Y_{\ell}^m(\theta, \phi).
\end{equation}

We take samples on the sphere at the locations $\set{(\theta_n,
\phi_n)}_{n=1}^N$. The $n$th sample is given by
\begin{equation}
  \label{eq:noisy_samples}
  \mu_n = f(\theta_n, \phi_n) + \varepsilon_n
\end{equation}
where $\varepsilon_n \sim {\cal N}(0, \sigma^2)$, and they are iid. By
$\bm{\zeta} = [\alpha_0, \ \theta_0, \ \phi_0]^\T$, we denote the vector of
parameters we estimate. To make the dependence on
$\bm{\zeta}$ explicit, we rewrite \eqref{eq:noisy_samples} slightly as
\begin{equation}
  \mu_n = f_n(\bm{\zeta}) + \varepsilon_n.
\end{equation}

With this notation in hand, we can write the conditional probability density
function of the $n$th measurement as
\begin{equation}
  p(\mu | \bm{\zeta}) = \frac{1}{\sqrt{2 \pi \sigma^2}} e^{-[\mu - f_n(\bm{\zeta})]^2 / (2\sigma^2)},
\end{equation}
so that the log-likelihood function is
\begin{align}
  L(\bm{\zeta}) &\bydef \ln p(\mu_1, \ldots, \mu_N | \bm{\zeta}) \nonumber \\
  &= \sum_{n=1}^N \left[-\tfrac{1}{2} \ln (2\pi\sigma^2) - (\mu_n\!-\!
    f_n(\bm{\zeta}))^2/(2\sigma^2) \right].
\end{align}
Consequently, differentiating $L$ with respect to any entry $\zeta_i$ of
$\bm{\zeta}$ gives
\begin{equation}
  \parder{L}{\zeta_i} = \frac{1}{\sigma^2} \sum_{n=1}^N \varepsilon_n \parder{f_n(\bm{\zeta})}{\zeta_i}.
\end{equation}

We can compute the three required derivatives
\begin{align*}
  &\parder{f_n(\bm{\zeta})}{\alpha_0} = \sum_{\ell=0}^{L-1} \sum_{\abs{m} \leq \ell}
   \overline{Y_{\ell}^m(\theta_0, \phi_0)}Y_{\ell}^m(\theta_n, \phi_n), && \\
  &\parder{f_n(\bm{\zeta})}{\theta_0} = \alpha_0 \sum_{\ell=0}^{L-1} \sum_{\abs{m} \leq \ell}
  \bigg[ m \cot \theta_0 \overline{Y_{\ell}^m(\theta_0, \phi_0)} && \\
  & + \sqrt{(l-m)(l+m+1)}e^{\I\phi_0} \overline{Y_{\ell}^{m+1}(\theta_0,\phi_0)} \bigg] Y_{\ell}^m(\theta_n, \phi_n), && \\
  &\parder{f_n(\bm{\zeta})}{\phi_0} = \alpha_0 \sum_{\ell=0}^{L-1} \sum_{\abs{m} \leq \ell}
   (-\I m) \overline{Y_{\ell}^m(\theta_0, \phi_0)}Y_{\ell}^m(\theta_n, \phi_n).&&
\end{align*}

Now $\nabla L = \left[ \parder{L}{\alpha_0}, \ \parder{L}{\theta_0},
  \ \parder{L}{\phi_0} \right]^\T$, and the Fisher information matrix is
\begin{flalign*}
  \mI(\bm{\zeta}) \bydef \expect \left[\nabla L(\bm{\zeta}) \nabla L(\bm{\zeta})^H\right]
  = \frac{1}{\sigma^2} \sum_{n=1}^N \nabla f_n(\bm{\zeta}) \nabla f_n(\bm{\zeta})^H.
\end{flalign*}
Let $\wh{\bm{\zeta}}$ be any unbiased estimator of the parameters $\bm{\zeta}$.
The CRLB can then be computed as

\begin{equation}
  \mathrm{cov}(\wh{\bm{\zeta}}) \succeq \mI(\bm{\zeta})^{-1}.
\end{equation}
%

\subsection{Rank of the annihilating matrix}

\label{appendix:rank_of_Z}

In this appendix, we show that the rank of the annihilating matrix $\mZ$
\eqref{eq:annihilating_matrix} is $K$ with probability one, as soon as it has
at least $K$ rows. It then follows follows that the annihilating filter $\vh$
is uniquely determined, up to a scaling factor, by solving $\mZ \vh = \vzero$.

Consider the factorization $\mDelta = \mX \mA \mU$,
\renewcommand{\arraycolsep}{0pt}
\begin{equation*}
  \mDelta = 
  \begin{bmatrix}
    1   &  \cdots & 1   \\
    x_1 &  \cdots & x_k \\
    \vdots &  &   \vdots \\
    x_1^{L-1} & \cdots & x_K^{L-1}
  \end{bmatrix}
  \begin{bmatrix}
    \alpha_1 & & \\
    & \ddots & \\
    & & \alpha_K
  \end{bmatrix}
  \begin{bmatrix}
    u_{1,1-L} & \cdots & 1 & \cdots & u_{1,L-1} \\
    \vdots & & & & \vdots \\
    u_{K,1-L} & \cdots & 1 & \cdots & u_{1,L-1}
  \end{bmatrix},
\end{equation*}
where $x_k = \cos \theta_k$ and $u_{k, m} = (\sin \theta_k)^\abs{m} e^{-\I m
\phi_k}$.

To construct the annihilating matrix $\mZ$ as in equation
\eqref{eq:annihilating_matrix}, we create Hankel blocks from columns of
$\mDelta$. The $(L-K) \times (K+1)$ Hankel block corresponding to the middle
($m=0$) column of $\mDelta$ can be factored as
\renewcommand{\arraycolsep}{2pt}
\begin{align}
  \mB_0 &= 
  \begin{bmatrix}
    x_1^{L-K-1} & \cdots & x_K^{L-K-1} \\
    \vdots &  &   \vdots \\
    x_1 &  \cdots & x_K \\
    1   &  \cdots & 1 
  \end{bmatrix}
  \begin{bmatrix}
    \alpha_1 & & \\
    & \ddots & \\
    & & \alpha_K
  \end{bmatrix}
  \begin{bmatrix}
    x_1^K & \cdots & x_1^0 \\
    \vdots  & & \vdots \\
    x_K^K & \cdots & x_K^0
  \end{bmatrix} \nonumber \\
  &\bydef \mX_{0} \mA \mXi.
\end{align}
The second block of the annihilating matrix obtained from the column
corresponding to $m=-1$ is similar,
\begin{equation}
  \mB_{-1} = \mX_{1} \mY_{-1} \mA  \mXi,
\end{equation}
where $\mY_{-1} \bydef \diag(u_{1,-1}, \ldots, u_{K,-1})$, and $\mX_m$ is
obtained by removing the $m$ leading rows from $\mX_0$. Then we can write
\begin{equation}
\begin{aligned}
  \mZ &= [\mB_0^\T,\ \mB_{-1}^\T,\ \mB_1^\T,\ \ldots,\ \mB_{K-L+1}^\T,\ \mB_{L-K-1}^\T]^\T \\
  &=
  \left[  
  \begin{array}{lclcl}
    \mX_{0} &\cdot& \mI &\cdot& \mA \mXi \\
    \mX_{1} &\cdot& \mY_{-1} &\cdot&\mA  \mXi \\
    \mX_{1} &\cdot& \mY_{1}  &\cdot& \mA  \mXi \\
      & & \vdots & & \\
    \mX_{L-K-1} &\cdot& \mY_{K-L+1} &\cdot & \mA  \mXi\\
    \mX_{L-K-1} &\cdot& \mY_{L-K-1} &\cdot & \mA  \mXi
  \end{array}
  \right],
\end{aligned}
\end{equation}
with the $\mA \mXi$ factor being common for all row-blocks. We want to show
that the nullspace of $\mZ$ has dimension one. To that end, we just need to
establish that the following matrix,
\begin{equation}
    \mT = 
    \left[
    \begin{array}{lcl}
    \mX_{0} &\cdot& \mI  \\
    \mX_{1} &\cdot& \mY_{-1}  \\
    \mX_{1} &\cdot& \mY_{1} \\
     & \vdots & \\
    \mX_{L-K-1} &\cdot& \mY_{K-L+1} \\
    \mX_{L-K-1} &\cdot& \mY_{L-K-1}
  \end{array}
  \right]
\end{equation}
has full column rank. To see why this is the case, let $\vv$ be a non-zero
vector such that $\vzero = \mZ \vv = \mT (\mA \mXi \vv)$. It then follows from
the full-rankness of $\mT$ that $\mA \mXi \vv = \vzero$. Since $\mA$ is a
diagonal matrix with non-zero entries on the diagonal and $\mXi$ is a $K
\times (K+1)$ Vandermonde matrix with distinct roots, the vector $\vv$ is
uniquely determined up to a multiplicative factor. We now show that the matrix
$\mT$ indeed has full column rank almost surely.

Any column in $\mT^\T$ is of the form
\begin{equation}
  \begin{bmatrix}
    (\cos \theta_1)^r (\sin \theta_1)^\abs{s} e^{\I \phi_1 s} \\
    \vdots \\
    (\cos \theta_K)^r (\sin \theta_K)^\abs{s} e^{\I \phi_K s}
  \end{bmatrix},
\end{equation}
where $0 \leq r < L - K - \abs{s}$ and $-(L-K-1) \leq s \leq L-K-1$. If the
locations of the Diracs are random, we can use the following lemma to show
that the matrix $\mT$ will have full column rank with probability one.

\begin{lem}
  \label{lem:almost_surely_fullrank}
  Draw $[\xi_k = (\theta_k, \phi_k)]_{k=1}^K$ independently at random from any
  absolutely continuous probability distribution on $\mathcal{R} = [0, \pi] \times [0,
  2\pi]$ (w.r.t. Lebesgue measure). Let $\mathcal{M} = \set{(r_1, s_1),
  \ldots, (r_N, s_N)}$ be a set of distinct integer pairs and let $\mG = [g_{pq}]$,
  where $g_{pq} = (\cos \theta_p)^{r_q} (\sin \theta_p)^\abs{s_q} e^{\I \phi_p
  s_q}$. Then $\mG$ has full rank almost surely.
\end{lem}

\begin{proof}
This proof is parallel to that of Theorem 3.2 from \cite{Bass:2005dr}. Let
$\mG_M$ be the upper left $M \times M$ minor of $\mG$. We define the
\emph{bad} set $\mathcal{B}_M$ as the set on which $\mG_M$ is singular,
\begin{equation}
  \mathcal{B}_M = \set{(\xi_1, \ldots, \xi_M) \in \mathcal{R}^M \ | \ \det \mG_M = 0}.
\end{equation}
The goal is to show that $\mu(\mathcal{B}_K) = 0$, where $\mu$ is the Lebesgue
measure on $\mathcal{R}^{K}$. We proceed by induction on $M$; for $M=1$, we have that
\begin{equation*}
  \mG_1 = [(\cos \theta_1)^{r_1} (\sin \theta_1)^\abs{s_1} e^{\I \phi_1 s_1}], 
\end{equation*}
which is non-zero almost surely, so the claim holds. Assume now that $M <
\min(K, N)$ and that the bad set $\mathcal{B}_M$ has measure zero. Let
$(\xi_1, \ldots, \xi_M) \notin \mathcal{B}_M$, \emph{i.e.}, $\mG_M$ is
invertible. Because it is invertible, there exists a unique coefficient vector
$\vb = \vb(\xi_1, \ldots, \xi_M)$ such that
\begin{equation}
  \label{eq:linear_combination_rank_of_Z}
  \mG_M \vb = \vg_{M+1},
\end{equation}
where by $\vg_{M+1}$ we denote the first $M$ entries of the last column of $\mG_{M+1}$.
The bigger matrix $\mG_{M+1}$ will be singular if and only if \emph{the same}
linear combination is also consistent with its $(M+1)$st row. In other
words, $\mG_{M+1}$ is invertible if and only if $\xi_{M+1}$ is not in the set
\begin{align*}
&\mathcal{Z}_{M}(\xi_1, \ldots, \xi_M) = \bigg\{ (\theta, \phi) = \xi \in \mathcal{R} \ \bigg| \\ 
& (\cos \theta)^{r_{M+1}} (\sin \theta)^\abs{s_{M+1}} e^{\I \phi s_{M+1}} \\ 
& \hspace{3cm} = \sum_{i=1}^M b_i (\cos \theta)^{r_i} (\sin \theta)^\abs{s_i} e^{\I \phi s_i}\bigg \}.
\end{align*}
For fixed $(\xi_1, \ldots, \xi_M)$, this is the set of zeros of a particular
(generalized) trigonometric polynomial, thus it has measure zero. Note that
the definition of $\mathcal{Z}_M$ makes sense only for $(\xi_1, \ldots, \xi_M)
\notin \mathcal{B}_M$, as otherwise $\mG_M$ is not invertible. Thus, the
solution $\vb$ to \eqref{eq:linear_combination_rank_of_Z} may not exist.

Consider now the following two sets:
\begin{equation*}
	\mathcal{U}_{M+1} \bydef \set{(\xi_1, \ldots, \xi_{M+1}) \ | \ (\xi_1, \ldots, \xi_M) \in \mathcal{B}_M, \xi_{M+1} \in \mathcal{R}}
\end{equation*}
and
\begin{equation*}
	\mathcal{V}_{M+1} \bydef \set{(\xi_1, \ldots, \xi_{M+1}) \ | \ (\xi_1, \ldots, \xi_M) \in \mathcal{R}^M, \xi_{M+1} \in \mathcal{Z}_{M}}.
\end{equation*}
The bad set $\mathcal{B}_{M+1}$ must be a subset of the set $\mathcal{U} \cup
\mathcal{V}$. But we just showed that the set $\mathcal{V}$ has measure zero;
by the induction hypothesis, $\mathcal{U}$ also has measure zero. Thus their
union, too, has measure zero.

It follows that $\mathcal{B}_{M+1}$ has measure zero. Finally, because the
distributions of $\xi_i$ are absolutely continuous w.r.t. the Lebesgue
measure, so is their product distribution. Hence the probability that $(\xi_1,
\ldots, \xi_K)$ lies in the zero-measure set $\mathcal{B}_K$ is zero.
\end{proof}

To complete the argument, note that the matrix $\mT^\T$ has the same form as
the matrix $\mG$ in the statement of Lemma \ref{lem:almost_surely_fullrank},
with $0 \leq r < L - K - \abs{s}$ and $-(L-K-1) \leq s \leq L-K-1$. Thus, the
columns of $\mT$ are independent with probability one, provided that its
number of rows is at least $K$.

\section*{Acknowledgement}

We thank Martin Vetterli for his support and mentoring.


\bibliographystyle{IEEEbib}
\bibliography{./spherical-fri-tsp}

\end{document}